\documentclass[journal, twoside]{IEEEtran}
\usepackage{amsmath}
\usepackage{amssymb, amsthm, epsfig, latexsym}
\usepackage{algorithm}
\usepackage{algpseudocode}
\usepackage{bm, bbm}
\usepackage{amsfonts, cases}
\usepackage{pstricks, pst-node, pst-text, pst-3d}
\usepackage{tabularx}
\usepackage{minibox}
\usepackage{booktabs}
\usepackage{dsfont}
\usepackage{verbatim}
\usepackage{mathrsfs}
\usepackage{units}
\usepackage{stmaryrd}
\usepackage{xspace}
\usepackage{enumerate}
\usepackage{hhline}

\usepackage[multiple]{footmisc}

\usepackage{cite}

\ifCLASSINFOpdf
\else
\fi
\usepackage{array}

\ifCLASSOPTIONcompsoc
  \usepackage[caption=false,font=normalsize,labelfont=sf,textfont=sf]{subfig}
\else
  \usepackage[caption=false,font=footnotesize]{subfig}
\fi

\def\cC{{\mathcal C}}
\def\cD{{\mathcal D}}

\def\cH{{\mathcal H}}

\def\cL{{\mathcal L}}
\def\cM{{\mathcal M}}

\def\cR{{\mathcal R}}

\def\cW{{\mathcal W}}

\def\sX{{\mathsf X}}
\def\sY{{\mathsf Y}}

\def\E{\mathbf{E}}

\def\I{{\mathrm I}}

\def\Reals{\mathbb{R}}

\def\BlkFn{\Lambda}

\def\GAM{p}
\def\NU{q}

\def\Bernoulli{\mathrm{Bern}}
\def\BSC{\mathrm{BSC}}
\def\BEC{\mathrm{BEC}}
\def\MIP{\mathrm{MIP}}
\def\BIC{\mathrm{BIC}}

\def\sm{{\mathsf m}}
\def\d{{\mathrm d}}
\def\deq{:=}

\def\eps{\varepsilon}

\def\Arikan{Ar{\i}kan}

\newcommand{\ATgood}{\mathop{\raisebox{0.25ex}{\footnotesize$\varoast$}}}
\newcommand{\ATbad}{\mathop{\raisebox{0.25ex}{\footnotesize$\boxast$}}}

\newtheorem{definition}{Definition}
\newtheorem{theorem}{Theorem}
\newtheorem{lemma}{Lemma}

\newtheorem{corollary}{Corollary}

\newtheorem{remark}{Remark}
\newtheorem{example}{Example}

\begin{document}
%
\title{Channel Polarization Through the Lens\\ of Blackwell Measures}
%
%
%

%

\author{Naveen Goela,~\IEEEmembership{Member,~IEEE} and Maxim Raginsky,~\IEEEmembership{Senior~Member,~IEEE}
\thanks{This work was presented in part at the \emph{IEEE International Symposium on Information Theory}, Barcelona, Spain on July 12, 2016, and at the \emph{Allerton Conference on Communication, Control, and Computing}, Illinois, USA on Sept. 28, 2016. This work was supported in part by the Center for Science of Information (CSoI), an NSF Science and Technology Center, under grant agreement CCF--0939370.}
\thanks{Maxim Raginsky is with the Coordinated Science Laboratory and with the Department of Electrical and Computer Engineering, University of Illinois at Urbana-Champaign, Champaign, IL 61801, USA. E-mail: \texttt{maxim@illinois.edu.}}
\thanks{N. Goela is a scientist at Tanium in Emeryville, CA. He was formerly with the Artificial Intelligence Lab at Technicolor, Palo Alto, CA, 94306, USA and the University of California, Berkeley, Berkeley, CA 94720-1770, USA. E-mail: \texttt{ngoela@alum.mit.edu.}}
}

%
%

\markboth{To appear in IEEE Transactions on Information Theory}{}

%



\maketitle

\begin{abstract}
Each memoryless binary-input channel ($\BIC$) can be uniquely described by its Blackwell measure, which is a probability distribution on the unit interval $[0,1]$ with mean $1/2$. Conversely, any such probability distribution defines a $\BIC$. The evolution of the Blackwell measure under \Arikan's polar transform is derived for general $\BIC$s, and is analogous to density evolution as cited in the literature. The present analysis emphasizes functional equations. Consequently, the evolution of a variety of channel functionals is characterized, including the symmetric capacity, Bhattacharyya parameter, moments of information density, Hellinger affinity, Gallager's reliability function, the Hirschfeld-Gebelein-R\'enyi maximal correlation, and the Bayesian information gain. The evolution of measure is specialized for symmetric $\BIC$s according to their decomposition into binary symmetric (sub)-channels ($\BSC$s), which simplifies iterative computations and the construction of polar codes. It is verified that, as a consequence of the Blackwell--Sherman--Stein theorem, all channel functionals $\I_f$ that can be expressed as an expectation of a convex function $f$ with respect to the Blackwell measure of a channel polarize in each iteration due to the polar transformation on the class of symmetric $\BIC$s. Moreover, for $f$ either convex or non-convex, a necessary and sufficient condition is established to determine whether the random process associated with each $\I_f$ is a martingale, submartingale, or supermartingale. Represented via functional inequalities in terms of $f$, this condition is numerically verifiable for all $\I_f$, and can generate analytical proofs. To exhibit one such proof, it is shown that the random process associated with the squared maximal correlation parameter is a supermartingale, and converges almost surely on the unit interval $[0,1]$.
\end{abstract}


\section{Introduction}
%
%
%
%
\IEEEPARstart{I}{ntroduced} in a seminal paper by E.~\Arikan~\cite{arikan09}, polar codes are a structured family of codes that provably achieve the capacity of symmetric, memoryless, binary-input channels ($\BIC$s) with low encoding and decoding complexity. The polar code transforms $N \deq 2^{n}$ independent copies of a symmetric $\BIC$ $W$ into $N$ polarized channels whose individual capacities approach either $0$ or $1$ with increasing block length $N$. The fraction of perfect channels among the $N$ transformed channels approaches $I(W)$, the symmetric capacity of $W$.

The original work of \Arikan\ establishes that the random process associated with $I(W)$ is a martingale, and that the random process associated with the Bhattacharyya parameter $Z(W)$ is a supermartingale~\cite{arikan09}. A central objective of the present paper is to analyze random processes associated with a broad class of channel functionals $\I_f(W)$ for $\BIC$s, including $I(W)$ and $Z(W)$, that can be expressed as an expectation of a function $f$ with respect to the Blackwell measure of $W$. Each channel $W$ is viewed through the lens of its Blackwell measure, which is a representation that originates in the work of Blackwell conducted in 1951 on the comparison of statistical experiments~\cite{BlackwellCE1,BlackwellCE2} (see also~\cite{torgersen1991comparison} for a modern synthesis).

\subsection{Overview of contributions}

The Blackwell measure is defined for memoryless $\BIC$s in Sec.~\ref{sec:RepresentationsOfBICs}. Due to the Blackwell--Sherman--Stein theorem~\cite{BlackwellCE1,BlackwellCE2,Sherman51,Stein51}, each memoryless $\BIC$ is uniquely specified by its Blackwell measure, which is a probability distribution on the unit interval $[0, 1]$ with mean $1/2$. Conversely, any such measure defines a memoryless $\BIC$. We discuss the relation of the Blackwell measure to other representations of channels, including the information density, $D$-distributions defined in the theory of density evolution~\cite{modern_coding_theory_book_2008}, and Neyman--Pearson regions defined in the theory of binary hypothesis testing (see e.g.,~\cite[Sec.~12.1 and 12.2]{Polyanskiy_Wu_lecturenotes}). The following list provides an overview of our main contributions:
\begin{itemize}
\item Sec.~\ref{sec:InducedFunctionalsOfBICs} presents the fact that any measurable function $f: [0,1] \rightarrow \Reals$ induces a functional $\I_f(W)$ of the channel $W$ through its Blackwell measure. Listed in Table~\ref{tbl:FunctionalsOfBICs}, several channel parameters can be expressed as $\I_f(W)$ for a particular choice of $f$. Non-trivial examples include the Hellinger affinity, moments of information density, Gallager's reliability function, the Hirschfeld-Gebelein-R\'enyi maximal correlation, and the Bayesian information gain.

\item Sec.~\ref{sec:OutputSymmetricBICs} develops the concept of channel decompositions for symmetric $\BIC$s, as introduced by Land and Huber~\cite{huber_2006}. Each symmetric $\BIC$ $W$ is equivalent to a compound channel, consisting of elementary subchannels that are $\BSC$s~\cite{huber_2006}. The Blackwell measure of a channel $W$ may be written in terms of the Blackwell measures of its subchannels. Consequently, the functional $\I_f(W)$ of $W$ may be computed in terms of the functionals of elementary subchannels of $W$.

\item In Sec.~\ref{sec:PolarTransformOfChannels}, \Arikan's polar transform is defined~\cite{arikan09}. The polar transform is applied to two $\BIC$s $W_1$ and $W_2$, and yields two transformed channels, denoted by $W_1 \ATbad W_2$ and $W_1 \ATgood W_2$. The Blackwell measures of $W_1 \ATbad W_2$ and $W_1 \ATgood W_2$ are derived in terms of the Blackwell measures of $W_1$ and $W_2$. While analogous to density evolution for polar codes~\cite{polar_code_density_evolution_2009}, our analysis emphasizes functional equations in proofs, and is applicable to asymmetric channels. For symmetric $\BIC$s, the Blackwell measures of $W_1 \ATbad W_2$ and $W_1 \ATgood W_2$ are derived explicitly by exploiting the fact that $W_1$ and $W_2$ are compound channels consisting of elementary subchannels (i.e., $\BSC$s). In Sec.~\ref{sec:SuccessiveChannelPolarization}, iterative computations are simplified to construct polar codes for arbitrary symmetric $\BIC$s.

\item In Sec.~\ref{sec:PolarizationBroadClassOfFunctionals}, it is verified that, as a consequence of the Blackwell--Sherman--Stein theorem~\cite{BlackwellCE1,BlackwellCE2,Sherman51,Stein51}, all channel functionals $\I_f(W)$ that can be expressed as an expectation of a convex function $f$ with respect to the Blackwell measure of $W$ polarize in each iteration of the polar transform on the class of symmetric $\BIC$s. This type of result has been independently discovered in the modern literature in both~\cite[Ch.~4]{modern_coding_theory_book_2008} and~\cite[Ch.~6]{alsan_phd14}.

\item Sec.~\ref{sec:PolarizationRandomProcesses} analyzes bounded random processes associated with channel functionals $\I_f(W)$ for symmetric $\BIC$s $W$, where $f$ may be either convex or non-convex. A necessary and sufficient condition is established to determine whether the random process associated with each $\I_f(W)$ is a martingale, submartingale, or supermartingale. Represented via functional inequalities in terms of $f$, this condition is numerically verifiable for all channel functionals $\I_f(W)$. Moreover, as we demonstrate, it provides a method for analytical proofs. Applying this method in Sec.~\ref{sec:NewSubSuperMartingales}, we prove that the random process associated with the squared maximal correlation parameter is a supermartingale, and converges almost surely on the unit interval $[0, 1]$.
\end{itemize}

\subsection{Relation to prior work}\label{subsec:RelatedPriorWork}

Beyond the symmetric capacity $I(W)$ and Bhattacharyya parameter $Z(W)$, more complex channel functionals have been studied. Alsan and Telatar prove that the random process corresponding to Gallager's reliability function $E_0$~\cite[Ch.~5.6]{gallager68}, which is related to various error exponents and cutoff rates~\cite{arikan_cutoff_rate_paper_2006}, is a submartingale~\cite[Theorem~2]{alsan_telatar_2014}~\cite[Theorem~4.7]{alsan_phd14}). Channel combining and splitting by iteratively applying \Arikan's polar transform increases and improves $E_0$. In~\cite[Theorem~1]{arikan_varentropy_2016}, \Arikan\ characterizes the evolution of the variance of the information density, called ``varentropy" generally, or ``channel dispersion" in the case of uniform input distribution to the channel. \Arikan\ proved that the varentropy decreases after each iteration of the polar transform. In the present paper, we prove that both Gallager's reliability function and the second moment of information density related to channel dispersion are induced functionals $\I_f(W)$ of $W$ for particular choices of $f$. As explained in Sec.~\ref{sec:PolarizationRandomProcesses} of the present paper (see e.g., Theorem~\ref{thm:fRelationsForOutputSymmetricBICs} and Corollary~\ref{corollary:NewMartingales}), it is feasible to verify and replicate these prior results, both numerically and analytically.

While the convergence of bounded martingales is well-known in mathematics, alternative methods exist to prove the convergence of random processes. For instance, it may be possible to relate one random process to an auxiliary random process whose convergence is rigorously established. Auxiliary processes were studied in~\cite{arikan09}. In more recent work, the authors of~\cite{alsan_telatar_simple_proof_2016} provide a simple proof of channel polarization that bypasses the explicit use of martingales. On a related topic, going beyond the concept of real-valued random processes, one can define a framework for \textit{channel-valued} random processes. Convergence in distribution, also called the ``weak convergence'' of channel densities, is defined for channels in the context of density evolution~\cite[Ch.~4]{modern_coding_theory_book_2008}. Subsequent to our work,~\cite{raj_nasser_2019} provides evidence for the convergence of \Arikan's \textit{channel-valued} random process in a topological space.

Since \Arikan's discovery of polarization, significant advances in theory have been made including: (i) multilevel and $q$-ary polarization~\cite{barg_qary_polar_2013,pradhan_multilevel_2013}; (ii) generalized $\ell \times \ell$ polarization matrices and algebraic constructions~\cite{korada_sasoglu_urbanke_2010,tanaka_2014}; (iii) refinements to the rate of polarization, scaling laws, and asymptotic analysis of polar codes~\cite{telatar_rate_polar_2009,hassani_rate_dependent_2013,hassani_finite_scaling_2014,guruswami_gap_to_capacity_2015,mondelli_finite_scaling_2016,sudan_stoc_GNRS_18}. Prior work focuses primarily on \Arikan's martingale corresponding to $I(W)$ and mutual information.

Beyond their application to point-to-point channels, polar codes have also been invented for several multi-user channels such as multiple-access channels~\cite{abble_mac_2012}, broadcast channels~\cite{goela_broadcast_2015}, and wiretap channels~\cite{vardy_wiretap_2011}. If the notions of symmetry and information combining~\cite{huber_2006} could be extended to multi-user channels, a corresponding measure-theoretic framework of polarization could be developed.

\subsection{Frequently used notation}\label{subsec:UsedNotation}

The following mathematical notations are adopted in the sequel. For $p,q \in [0,1]$, we let $\bar{p} \deq 1 - p$ (for $p \in \{0,1\}$, this is the Boolean {\sc not}) and $p \star q \deq p\bar{q} + \bar{p}q$.  For $a,b \in \Reals$, $a \wedge b \deq \min(a,b)$ and $a \vee b \deq \max(a,b)$. The closure of a set $S$ is denoted by ${\rm cl}\{S\}$. For a random variable $X$ defined on a probability space $(\Omega, \mathcal{F}, P)$, the probability law of $X$ is the induced measure of $X$ under $P$ on $\Reals$~\cite[Ch.~2]{measure_theory_athreya_2006}. We will denote by $\cL(X)$ the probability law of $X$. The notation $\delta_x$ denotes the Dirac measure centered on a fixed point $x$ in a measurable space. The binary entropy function is denoted by $h_2(x) := -x\log_2(x) - \bar{x} \log_2 \bar{x}$ for $x \in [0,1]$. More generally, we also define $\psi_{r}(x) := x(1 + \log_2 x)^{r} + \bar{x} (1 + \log_2 \bar{x} )^{r}$ for $r$ a positive integer, and $x \in [0,1]$.

\section{Representations of BICs}\label{sec:RepresentationsOfBICs}

In this work, we focus on discrete, memoryless $\BIC$s with finite output alphabets:

\begin{definition}[Discrete, memoryless, binary-input channel ($\BIC$)]
A discrete, memoryless, binary-input channel ($\BIC$) is a pair $(\sY,W)$, where $\sY$ is the finite output alphabet and $W = \big(W(\cdot|0), W(\cdot|1)\big)$ is a pair of probability distributions on $\sY$. For $x \in \{0,1\}$, $W(\cdot|x)$ is the probability distribution of the channel output when the channel input is equal to $x$.
\end{definition}

\noindent The \textit{channel transition matrix} is the most familiar representation of a $\BIC$:

\begin{definition}[Channel transition matrix] For a $\BIC$ $(\sY,W)$, let $T_{W}$ denote the $2 \times |\sY|$ matrix whose elements are $W(y|x)$ for $(x,y) \in \{0,1\} \times \sY$. \end{definition}

\begin{example}[Binary erasure channel $\BEC(\varepsilon)$] {\em The binary erasure channel with erasure probability $\varepsilon$ is a $\BIC$ $(\sY,W)$ with $\sY = \{0, 1, \texttt{e}\}$, $W(\cdot|0) = \bar{\eps} \delta_0 + \eps \delta_{\texttt{e}}$, and $W(\cdot|1) = \bar{\eps}\delta_1 + \eps \delta_{\texttt{e}}$.
The transition matrix is
\begin{align}
T_{\BEC(\varepsilon)} & := \left[ \begin{array}{ccc} 1-\varepsilon & 0 & \varepsilon \\ 0 & 1 - \varepsilon & \varepsilon \end{array}\right]. \notag
\end{align}}
\end{example}

\begin{example}[Binary symmetric channel $\BSC(\GAM)$] {\em The binary symmetric channel with bit-flip probability $\GAM$ is a $\BIC$ $(\sY,W)$ with $\sY = \{0, 1\}$, $W(\cdot|0) = {\rm Bern}(p)$, and $W(\cdot|1) = {\rm Bern}(\bar{p})$. The
transition matrix is
\begin{align}
T_{\BSC(\GAM)} & := \left[ \begin{array}{cc} 1-\GAM & \GAM \\ \GAM & 1-\GAM \end{array}\right]. \notag
\end{align}}
\end{example}
\noindent In the remainder of this section, we describe a number of alternative representations of $\BIC$s that will be used in the sequel.

\subsection{The Blackwell measure}

The Blackwell measure~\cite{BlackwellCE1,BlackwellCE2,torgersen1991comparison} particularized to $\BIC$s is defined as the distribution of the posterior probability of the binary input being $0$, assuming a uniform input distribution to the channel:

\begin{definition}[Blackwell measure of a $\BIC$]\label{def:BlackwellMeasureForBIC}
Given a $\BIC$ $(\sY,W)$, let $(X,Y)$ be a random couple on $\{0,1\} \times \sY$ with $P_X = \Bernoulli(1/2)$ and $P_{Y|X} = W$. Define the function
\begin{align}
\BlkFn_W(y) & := \frac{ W(y|0) }{ W(y|0) + W(y|1) }. \label{eqn:BlackwellMeasureBIC}
\end{align}
The random variable $S = \BlkFn_W(Y)$, which is equal to the posterior probability of $X=0$ given $Y$, takes values in the unit interval $[0,1]$ and has mean $1/2$. The Blackwell measure of $W$, which we will denote by $\sm_W$, is the probability law of random variable $S$ (see~\cite[Ch.~2]{measure_theory_athreya_2006} for the definition of probability law).
\end{definition}

\begin{example}[Blackwell measures for $\BEC(\varepsilon)$ and $\BSC(\GAM)$] {\em The Blackwell measures for the $\BEC(\varepsilon)$ and $\BSC(\GAM)$ are
\begin{align}
\sm_{\BSC(\GAM)} & = \frac{1}{2}\delta_{\GAM} + \frac{1}{2}\delta_{\bar{\GAM}}, \label{eqn:BSCDeltaMeasure} \\
\sm_{\BEC(\varepsilon)} & = \frac{\bar{\varepsilon}}{2}\delta_0 + \frac{\bar{\varepsilon}}{2}\delta_1 + \varepsilon\delta_{1/2}. \label{eqn:BECDeltaMeasure}
\end{align}}
\end{example}

Given two $\BIC$s $(\sY,W)$ and $(\sY',W')$, we say that $W$ dominates $W'$ (or is more informative than $W'$) in the sense of Blackwell \cite{BlackwellCE1,BlackwellCE2,torgersen1991comparison} if there exists a random transformation $K$ from $\sY$ to $\sY'$, such that $W' = K \circ W$, i.e., for all $x \in \{0,1\}$ and all $y' \in \sY$,
$$
W'(y'|x) = \sum_{y \in \sY} K(y'|y) W(y|x),
$$
In other words, $W$ dominates $W'$ exactly when it is stochastically degraded with respect to $W$. In that case, we write $W \succeq W'$. We say that $W$ and $W'$ are equivalent if $W \succeq W'$ and $W' \succeq W$. In that case, we write $W \equiv W'$. The fundamental nature of the Blackwell measure is evident from the following theorem:

\begin{theorem}[Blackwell--Sherman--Stein~\cite{BlackwellCE1,BlackwellCE2,Sherman51,Stein51}]\label{thm:BSS} Consider two $\BIC$s $W$ and $W'$. Then:
	\begin{enumerate}
		\item $W \equiv W'$ if and only if $\sm_W = \sm_{W'}$ (that is, the Blackwell measure specifies the channel uniquely up to equivalence). Moreover, let $\cM$ denote the collection of all Borel probability measures on $[0,1]$ with mean $1/2$. Then for any $\sm \in \cM$ there exists a $\BIC$ $W$, unique up to equivalence, such that $\sm = \sm_W$.\footnote{The $\BIC$ $W$ has a finite output alphabet if and only if $\sm_W$ has finite support. This is precisely the setting of this paper.}\footnote{The Blackwell measure is also defined in the context of hidden Markov models~\cite{BlackwellCE3}.}
		\item $W \succeq W'$ if and only if
		$$
		\int_{[0,1]} f \d\sm_W \ge \int_{[0,1]} f\d\sm_{W'}
		$$
		for every convex $f : [0,1] \to \Reals$.
	\end{enumerate}
\end{theorem}

\begin{remark} {\em There are one-to-one correspondences between the Blackwell measure $\sm_W$ and other probabilistic objects associated to the $\BIC$ $W$, such as its $L$- and $D$-distributions~\cite[Ch.~4]{modern_coding_theory_book_2008}. The Blackwell measure is also a special case of the so-called $\alpha$-representation\cite{arikan_varentropy_2016}.}
\end{remark}

\subsection{$D$-distributions and density evolution}

Recall the following definition of the $D$-distribution:

\begin{definition}[$D$-distribution~\cite{modern_coding_theory_book_2008}] Given a $\BIC$ $(\sY,W)$, let $(X,Y)$ be a random couple on $\{0,1\} \times \sY$ with $P_X = \Bernoulli(1/2)$ and $P_{Y|X} = W$. Define the function
\begin{align}
\Delta_{W}(y) & := P_{X|Y}(0|y) - P_{X|Y}(1|y), \label{eqn:DdistributionDef}
\end{align}
and let $D = \Delta_{W}(Y)$. Then, the $D$-distribution is the cumulative distribution function (c.d.f.) of $D$ conditioned on the event $X = 0$~\cite[Ch.~4]{modern_coding_theory_book_2008}.
\end{definition}

\noindent It follows from Eqn.~\eqref{eqn:BlackwellMeasureBIC} and Eqn.~\eqref{eqn:DdistributionDef} that for arbitrary $y \in \sY$,
\begin{align}
\Delta_{W}(y) & = \frac{W(y|0) - W(y|1)}{W(y|0)+W(y|1)} \notag \\
&\equiv 2\BlkFn_{W}(y)-1. \label{eqn:DeltaFunctionAlternativeForm}
\end{align}
Thus, the probability law of $\Delta_W(Y) = 2S-1$ also specifies $W$ uniquely up to Blackwell equivalence. Moreover, $\Delta_W(Y)$ takes values in $[-1,1]$, has mean zero, and is symmetric, i.e., $\Delta_W(Y)$ and $-\Delta_W(Y)$ have the same probability law.

\subsection{Information density}

Information density furnishes another useful description of $\BIC$s. Let $(X,Y)$ be a random couple taking values in a finite product space $\sX \times \sY$. The information density is defined as
\begin{align}
i(x;y) := \log_2 \frac{P_{Y|X}(y|x)}{P_{Y}(y)}, \label{eqn:IDensity}
\end{align}
where $P_Y(y) = \sum_{x \in \sX} P_{X}(x)P_{Y|X}(y|x)$. The expectation and the variance of the information density are the mutual information and the information variance:
\begin{align}
I(X; Y) & = \mathbf{E}\left[i(X; Y)\right], \notag \\
V(X;Y) & = \mathbf{E}\left[i^{2}(X;Y)\right] - (I(X;Y))^{2}. \notag
\end{align}
We particularize this to $\BIC$s with equiprobable inputs:

\begin{definition}[Information density for a $\BIC$]\label{def:InformationDensity} Given a $\BIC$ $(\sY,W)$, let $(X,Y)$ be a random couple on $\{0,1\} \times \sY$ with $P_X = \Bernoulli(1/2)$ and $P_{Y|X} = W$. The information density of $(X,Y)$ is given by
\begin{align}
i_W(x; y) = \log_2 \frac{ W(y|x) } { \frac{1}{2} W(y|0) + \frac{1}{2}W(y|1) }. \label{eqn:IDensityBIC}
\end{align}
\end{definition}
The expectation and variance of $i_W(X; Y)$ with $X \sim \Bernoulli(1/2)$ are known as the symmetric capacity $I(W)$ and symmetric dispersion $V(W)$, respectively. To express these parameters succinctly, we introduce the $r$th  moment of the information density:
\begin{align}\label{eqn:RthMomentInformationDensity}
M_r(W) := \E[i_W(X;Y)^r].
\end{align}
Then $I(W) = M_1(W)$ and $V(W) = M_2(W) - I^2(W)$.
From Eqn.~\eqref{eqn:BlackwellMeasureBIC} and Eqn.~\eqref{eqn:IDensityBIC}, it follows that
$$
i_{W}(x; y) =
\begin{cases}
1 + \log_2 \BlkFn_{W}(y) & \mbox{if} ~ x = 0; \\
1 + \log_2 (1-\BlkFn_{W}(y)) & \mbox{if} ~ x = 1;
\end{cases}
$$
for arbitrary $y \in \sY$. Therefore, the information density specifies $W$ uniquely up to Blackwell equivalence.

\subsection{The Neyman--Pearson region}

Another useful representation of $\BIC$s arises from the theory of binary hypothesis testing (see~\cite[Sec.~12.1 and 12.2]{Polyanskiy_Wu_lecturenotes}).

\begin{definition}[Neyman--Pearson region]\label{def:NPRegion} For a $\BIC$ $(\sY,W)$,  the \textit{Neyman--Pearson region} $\cR_{{\rm NP}}(W)$ is a subset of $[0,1]^2$ consisting of all points $(\xi, \eta)$, for which there exists some $f : \sY \to [0,1]$, such that
\begin{align}
\xi & = \sum_{y \in \sY} f(y) W(y|0), \label{eq:NP_pointsA} \\
\mbox{and}~~\eta & = \sum_{y \in \sY} f(y) W(y|1) \label{eq:NP_pointsB}.
\end{align}
\end{definition}

\noindent The Neyman--Pearson region has the following properties:
\begin{enumerate}[i)]
\item It is a closed and convex subset of $[0,1]^2$.
\item It contains the diagonal $\cD \deq \{(\xi,\xi) : \xi \in [0,1]\}$.
\item It is equal to the closed convex hull of all points $(\xi,\eta)$, where $\xi$ has the form~\eqref{eq:NP_pointsA} and $\eta$ has the form~\eqref{eq:NP_pointsB}, with $f$ taking values in $\{0,1\}$:
\begin{IEEEeqnarray*}{rCl}
\IEEEeqnarraymulticol{3}{l} { \cR_{{\rm NP}}(W) } \\
\quad \quad & = & {\rm cl}\left\{{\rm conv}\Big\{ \left(W(A|0), W(A|1) \right) : A \subseteq \sY \Big\}\right\},
\end{IEEEeqnarray*}
where $W(A|x) \deq \sum_{y \in A} W(y|x)$.
\end{enumerate}
The following fundamental result is a consequence of the Blackwell--Sherman--Stein theorem:
\begin{theorem}[Neyman--Pearson criterion for Blackwell dominance]\label{thm:NeymanPearsonCriterion} Consider two $\BIC$s $W$ and $W'$. The Neyman--Pearson criterion for Blackwell dominance is given by,
$$
W \succeq W' ~ \Longleftrightarrow ~ \cR_{{\rm NP}}(W) \supseteq \cR_{{\rm NP}}(W').
$$
\end{theorem}
\noindent For example, one can show that
\begin{IEEEeqnarray*}{rCl}
\IEEEeqnarraymulticol{3}{l} { \cR_{{\rm NP}}(\BSC(p)) } \\
\quad \quad & = & {\rm cl}\Bigl\{{\rm conv} \left\{ (0,0), (p,\bar{p}), (\bar{p},p), (1,1) \right\}\Bigl\}.
\end{IEEEeqnarray*}
Then $\cR_{{\rm NP}} (\BSC(0)) = [0,1]^2$, and $\cR_{{\rm NP}}(\BSC(1/2)) = \cD$.

\section{Functionals of $\BIC$s}\label{sec:InducedFunctionalsOfBICs}

Any measurable function $f : [0,1] \to \Reals$ induces a functional $\I_f$ on the collection of all $\BIC$s via
\begin{IEEEeqnarray}{rCl}
\I_f(W) & \deq & \int_{[0,1]} f \d\sm_W = \E[f(S)], \label{eqn:FunctionalBICEqnDefinition}
\end{IEEEeqnarray}
where $S \sim \sm_W$. As summarized in Table~\ref{tbl:FunctionalsOfBICs} and explained in detail in this section, a variety of channel characteristics can be expressed in this way.

\begin{table*}[!t]
\centering
\caption{Real-valued functionals of $\BIC$s}
\label{tbl:FunctionalsOfBICs}
\begin{IEEEeqnarraybox}[\IEEEeqnarraystrutmode\IEEEeqnarraystrutsizeadd{2pt}{1pt}]{x/l/v/l/v/l/x}
\IEEEeqnarrayrulerow[1.42pt]\\
& {\rm \textbf{Measurable~Function} } &&{\rm \textbf{Induced~Functional} } & & {\rm \textbf{English~Description} } &\\
& f : [0,1] \to \Reals &&  \I_f(W)  & &  & \\
\IEEEeqnarraydblrulerow\\
\IEEEeqnarrayseprow[3pt]\\
&  f(s) = 1 - h_2(s)   &&     I(W)    & & {\rm Mutual~Information} & \IEEEeqnarraystrutsize{0pt}{0pt}\\
\IEEEeqnarrayseprow[3pt]\\
\IEEEeqnarrayrulerow\\
\IEEEeqnarrayseprow[3pt]\\
&  f(s) = s(1 + \log_2 s)^{r} + \bar{s} (1 + \log_2 \bar{s} )^{r}, ~r\in \mathbb{Z}^{+}   &&     M_{r}(W) & & {\rm Moments~of~Information~Density} & \IEEEeqnarraystrutsize{0pt}{0pt}\\
\IEEEeqnarrayseprow[3pt]\\
\IEEEeqnarrayrulerow\\
\IEEEeqnarrayseprow[3pt]\\
&  f(s) = 2\sqrt{s(1-s)}   &&     Z(W)    & & {\rm Bhattacharyya~Parameter} & \IEEEeqnarraystrutsize{0pt}{0pt}\\
\IEEEeqnarrayseprow[3pt]\\
\IEEEeqnarrayrulerow\\
\IEEEeqnarrayseprow[3pt]\\
&  f(s) = 2s^\alpha(1-s)^{1-\alpha}, ~\alpha \in (0,1)   &&     \cH_\alpha(W)    & & {\rm Hellinger~Affinity} & \IEEEeqnarraystrutsize{0pt}{0pt}\\
\IEEEeqnarrayseprow[3pt]\\
\IEEEeqnarrayrulerow\\
\IEEEeqnarrayseprow[3pt]\\
&  f(s) = 2^{-\rho}\left(s^{\frac{1}{1+\rho}} + (1-s)^{\frac{1}{1+\rho}} \right)^{1+\rho}, ~\rho\geq 0 && \exp\left(-E_0(\rho,W)\right) & & {\rm Gallager's~Function}~E_0(\rho,W)~ & \IEEEeqnarraystrutsize{0pt}{0pt}\\
\IEEEeqnarrayseprow[3pt]\\
\IEEEeqnarrayrulerow\\
\IEEEeqnarrayseprow[3pt]\\
&  f(s) = \bar{\lambda} \wedge \lambda - (2\bar{\lambda} s) \wedge (2\lambda\bar{s}), ~\lambda \in [0,1] && B_\lambda(W) & & {\rm Bayesian~Information~Gain} & \IEEEeqnarraystrutsize{0pt}{0pt}\\
\IEEEeqnarrayseprow[3pt]\\
\IEEEeqnarrayrulerow\\
\IEEEeqnarrayseprow[3pt]\\
&  f(s) = |2s-1| && 1-2P_{\rm{e,ML}}(W) & & {\rm Maximum~Likelihood~Decoding~Error}~P_{\rm{e,ML}}(W) & \IEEEeqnarraystrutsize{0pt}{0pt}\\
\IEEEeqnarrayseprow[3pt]\\
\IEEEeqnarrayrulerow\\
\IEEEeqnarrayseprow[3pt]\\
&  f(s) = (2s-1)^2 && \rho^{2}_{\rm{max}}(W) & & {\rm Squared~Maximal~Correlation} & \IEEEeqnarraystrutsize{0pt}{0pt}\\
\IEEEeqnarrayseprow[3pt]\\
\IEEEeqnarrayrulerow[1.42pt]\\
\end{IEEEeqnarraybox}
\end{table*}

\subsection{Symmetric capacity $I(W)$} With $f(s) = 1 - h_2(s)$, where $h_2(\cdot)$ is the binary entropy function, $\I_f(W)$ is equal to the \textit{symmetric capacity} $I(W)$ of $W$ \cite{arikan09}, i.e., the mutual information of $W$ with uniform input distribution. Indeed, let $(X,Y)$ be a random couple with $X \sim \Bernoulli(1/2)$ and $P_{Y|X} = W$. Then
\begin{align*}
\I_f(W) &= 1-\E[h_2(S)] \\
&=1 + \sum_{x \in \{0, 1\}} \E\left[ \frac{W(Y|x)}{2P_Y(Y)} \log_2 \frac{W(Y|x)}{2P_Y(Y)}
\right] \\
&=1+ \E\left[\E\Bigg[\log_2\frac{W(Y|X)}{2P_Y(Y)}\Bigg|X\Bigg]\right] \\
&= D(P_{Y|X}\|P_Y|P_X) \\
&\equiv I(W).
\end{align*}

\subsection{The $r$th moment of information density $M_r(W)$}\label{subsec:RthMomentInformationDensity} Let $r$ be a positive integer. If we take
\begin{IEEEeqnarray}{rClCl}
f(s) & = & \psi_{r}(s) & \deq & s(1 + \log_2 s)^{r} + \bar{s} (1 + \log_2 \bar{s} )^{r}, \nonumber
\end{IEEEeqnarray}
then $\I_f(W)$ is equal to the $r$-th moment of the information density $M_r(W)$, assuming uniform input distribution, as defined in Eqn.~\eqref{eqn:RthMomentInformationDensity}. The following equalities establish our claim:
\begin{IEEEeqnarray}{rRl}
\IEEEeqnarraymulticol{3}{l} { M_r(W) } \nonumber \\
~ & := & \mathbf{E}\left[\left(i_{W}(X; Y)\right)^{r}\right] \nonumber \\
~ &  = & \mathbf{E}\left[ \left(  \log_2 \frac{ 2W(Y|X) }{W(Y|0) + W(Y|1)}  \right)^{r} \right] \nonumber \\
~ &  = & \mathbf{E}\left[ \left(1 + \log_2 \frac{ W(Y|X) }{W(Y|0) + W(Y|1)}  \right)^{r} \right] \nonumber  \\
~ &  = & \mathbf{E} \left[ P_{X|Y}(0|Y) \left(1 + \log_2 \frac{ W(Y|0) }{W(Y|0) + W(Y|1)}  \right)^{r} \right]  \nonumber \\
~ &  ~ & ~ + ~ \mathbf{E} \left[ P_{X|Y}(1|Y)\left(1 + \log_2 \frac{ W(Y|1) }{W(Y|0) + W(Y|1)}  \right)^{r} \right] \nonumber \\
~ &  = & \mathbf{E} \left[ \BlkFn_{W}(Y) \left(1 + \log_2 \BlkFn_{W}(Y)  \right)^{r} \right]  \nonumber \\
~ &  ~ & ~ + ~ \mathbf{E} \left[ (1 - \BlkFn_{W}(Y)) \left(1 + \log_2 (1 - \BlkFn_{W}(Y)) \right)^{r} \right] \nonumber \\
~ &  = & \mathbf{E} \left[ S(1 + \log_2 S)^{r} + \bar{S}(1 + \log_2 \bar{S})^{r} \right] \nonumber \\
~ &  = & \mathbf{E} \left[ \psi_r(S) \right] \notag \\
~ & \equiv & \I_f(W). \nonumber
\end{IEEEeqnarray}
Note that $\psi_1(s) = 1 - h_2(s)$ and $M_1(W) = I(W)$. The channel dispersion parameter $V(W)$ is defined as the variance of the information density, $V(W) \deq M_2(W) - (I(W))^{2}$, assuming a uniform input distribution. While the dispersion $V(W)$ cannot be expressed explicitly as an induced functional of the form $\I_f(W)$ for any $f : [0,1] \to \Reals$, we can use the variational representation of the variance as follows:
\begin{IEEEeqnarray*}{rCl}
V(W) & = & {\rm Var}[i(X;Y)] \\
     & = & \min_{c \in \Reals} \E[(i(X;Y)-c)^2] \\
     & = & \min_{c \in \Reals} \E[(1-h_2(S)-c)^2] \\
     & = & \min_{c \in \Reals} \E[(h_2(S)-c)^2].
\end{IEEEeqnarray*}
Thus, if we consider the family of functions $f_c(s) \deq (h_2(s) - c)^2$, $c \in \Reals$, we see that $V(W)$ can be expressed as
\begin{IEEEeqnarray*}{rCl}
V(W) & = & \min_{c \in \Reals} \I_{f_c}(W).
\end{IEEEeqnarray*}

\subsection{Hellinger affinity $\cH_\alpha(W)$} If we select $f(s) = 2s^\alpha(1-s)^{1-\alpha}$ for $\alpha \in (0,1)$, then $\I_f(W)$ is equal to the \textit{Hellinger affinity} of order $\alpha$:
\begin{IEEEeqnarray*}{rCl}
\cH_\alpha(W) & \deq & \sum_{y \in \sY} W(y|0)^\alpha W(y|1)^{1-\alpha}.
\end{IEEEeqnarray*}
Indeed, one can write,
\begin{IEEEeqnarray*}{rCl}
\I_f(W) & = & 2\, \E[S^\alpha(1-S)^{1-\alpha}] \\
& = & 2\,\E\left[\left(\frac{W(Y|0)}{2P_Y(Y)}\right)^\alpha \left(\frac{W(Y|1)}{2P_Y(Y)}\right)^{1-\alpha}\right] \\
& = & \E\left[\frac{1}{P_Y(Y)} W(Y|0)^\alpha W(Y|1)^{1-\alpha}\right] \\
& = & \sum_{y \in \sY} W(y|0)^\alpha W(y|1)^{1-\alpha} \\
& \equiv & \cH_\alpha(W).
\end{IEEEeqnarray*}
In particular, if we set $\alpha = 1/2$, then we recover the \textit{Bhattacharyya parameter}~\cite{arikan09},
\begin{IEEEeqnarray*}{rCl}
\cH_{1/2}(W) & = & Z(W) \deq \sum_{y \in \sY} \sqrt{W(y|0)W(y|1)}.
\end{IEEEeqnarray*}

\subsection{Gallager's function $E_0(\rho,W)$}

Gallager's $E_0$ function of a $\BIC$ $(\sY,W)$ with input distribution $P_X$ is defined as follows \cite[Ch.~5.6]{gallager68},
\begin{IEEEeqnarray*}{rCl}
E_0(\rho,P_{X},W) & \deq & -{\ln \sum_{y \in \sY} \sum_{x \in \{0,1\}} \left[P_X(x) W(y|x)^{\frac{1}{1+\rho}} \right]^{1+\rho}}, \nonumber
\end{IEEEeqnarray*}
for any $\rho \ge 0$. In particular, we define
\begin{IEEEeqnarray}{rRl}
E_0(\rho,W) & \deq & E_0(\rho,\Bernoulli(1/2),W) \nonumber \\
& = & -{\ln \sum_{y \in \sY} \left( \frac{1}{2}W(y|0)^{\frac{1}{1+\rho}} + \frac{1}{2}W(y|1)^{\frac{1}{1+\rho}}\right)^{1+\rho}}. \nonumber
\end{IEEEeqnarray}
Choosing $f$ as
\begin{IEEEeqnarray*}{rCl}
f(s) & = & 2^{-\rho}\left(s^{\frac{1}{1+\rho}} + (1-s)^{\frac{1}{1+\rho}} \right)^{1+\rho}
\end{IEEEeqnarray*}
yields an induced functional
\begin{IEEEeqnarray*}{rCl}
\I_f(W) & = & \exp\left(-E_0(\rho,W)\right).
\end{IEEEeqnarray*}
To see this, consider the following chain of equalities:
\begin{IEEEeqnarray}{rRl}
\IEEEeqnarraymulticol{3}{l} {  \I_f(W) } \nonumber \\
~ & := & \mathbf{E}\left[f(S)\right] \nonumber \\
~ &  = & \mathbf{E}\left[ 2^{-\rho} \left( S^{\frac{1}{1+\rho}} + (1-S)^{\frac{1}{1+\rho}} \right)^{1 + \rho} \right] \nonumber \\
~ &  = & \mathbf{E}\left[ 2^{-\rho} \left( (\BlkFn_{W}(Y))^{\frac{1}{1+\rho}} + (1-\BlkFn_{W}(Y))^{\frac{1}{1+\rho}} \right)^{1 + \rho} \right] \nonumber \\
~ & =  & \mathbf{E}\left[ 2^{-\rho} \left( \frac{ (W(Y|0))^{\frac{1}{1 + \rho}} + (W(Y|1))^{\frac{1}{1 + \rho}}} { (W(Y|0) + W(Y|1))^{\frac{1}{1 + \rho}} } \right)^{1 + \rho} \right]  \nonumber \\
~&  =  & \mathbf{E}\left[ 2^{-\rho}  \frac{ \left( (W(Y|0))^{\frac{1}{1 + \rho}} + (W(Y|1))^{\frac{1}{1 + \rho}} \right)^{1 + \rho} } { W(Y|0) + W(Y|1) }  \right]  \nonumber \\
~&  =  & \mathbf{E}\left[  \frac{ \left( \frac{1}{2}(W(Y|0))^{\frac{1}{1 + \rho}} + \frac{1}{2}(W(Y|1))^{\frac{1}{1 + \rho}} \right)^{1 + \rho} } { \frac{1}{2}W(Y|0) + \frac{1}{2}W(Y|1) }  \right]  \nonumber \\
~ & =  & \sum_{y \in \sY} \left( \frac{1}{2}(W(y|0))^{\frac{1}{1 + \rho}} + \frac{1}{2}(W(y|1))^{\frac{1}{1 + \rho}} \right)^{1 + \rho} \nonumber \\
~ & \equiv & \exp\left(-E_0(\rho,W)\right). \nonumber
\end{IEEEeqnarray}

\subsection{Bayesian information gain $B_{\lambda}(W)$}

Given a $\BIC$ $(\sY,W)$ and $\lambda \in [0,1]$, consider a random couple $(X,Y)$ with $X \sim \Bernoulli(\lambda)$ and $P_{Y|X} = W$. Define the \text{minimum Bayes risk}
\begin{IEEEeqnarray}{rRrl}
b_\lambda(W) & \deq & \IEEEeqnarraymulticol{2}{l} {\min_{g : \sY \to \{0,1\}} {\mathbf P}[g(Y) \neq X] } \nonumber \\
& = & \min_{g : \sY \to \{0,1\}} \Biggl( & \bar{\lambda} \sum_{y \in \sY}W(y|0) {\bf 1}_{\{g(y) = 1\}}  \nonumber \\
& & \quad +\> & \lambda \sum_{y \in \sY} W(y|1) {\bf 1}_{\{g(y) = 0\}} \Biggl). \nonumber
\end{IEEEeqnarray}
where the minimum is over all deterministic decoders $g : \sY \to \{0,1\}$. The \textit{Bayesian information gain} is defined as
\begin{IEEEeqnarray}{rRl}
B_\lambda(W) & \deq & b_\lambda(\BSC(1/2)) - b_\lambda(W). \nonumber \\
\noalign{\vspace{1.5\jot} \noindent We claim that
\vspace{1.5\jot}}
B_\lambda(W) & = & \I_{f_\lambda}(W), \label{eqn:BayesianInformationClaim}
\end{IEEEeqnarray}
with $f_\lambda(s) \deq \bar{\lambda} \wedge \lambda - (2\bar{\lambda} s) \wedge (2\lambda\bar{s})$. The proof of this claim is given in Appendix~\ref{app:BayesianInformationGainProof}. Moreover, since any convex $f : [0,1] \to \Reals$ can be approximated by a positive affine combination of such $f_\lambda$'s \cite{torgersen1991comparison}, it follows that $W \succeq W'$ if and only if $B_\lambda(W) \ge B_\lambda(W')$ for all $\lambda \in (0,1)$.


\subsection{Squared maximal correlation $\rho^{2}_{\rm{max}}(W)$}\label{subsec:SquaredMaxiximalCorrelation}

Given jointly distributed, real-valued random variables $(X,Y)$, the Hirschfeld-Gebelein-R\'enyi maximal correlation is defined as follows~\cite{hirschfeld_1935,gebelein_1941,renyi_max_corr_1935}:
\begin{align}
\rho_{\rm{max}}(X, Y) & := \sup_{(g(X),~h(Y)) \in \mathcal{S} } \E\left[ g(X)h(Y) \right],
\end{align}
where $g, h$ are real-valued functions, and $\mathcal{S}$ is the collection of pairs of real-valued random variables $(g(X), h(Y))$ such that
\begin{IEEEeqnarray}{rClCrr}
\E [g(X)] & = & \E [h(Y)] & = & ~0, & \quad \mbox{and} \nonumber \\
\E \left[g^{2}(X)\right] & = & \E \left[ h^{2}(Y)\right] & = & ~1. & \nonumber
\end{IEEEeqnarray}
The maximal correlation is bounded: $0 \leq \rho_{\rm{max}}(X, Y) \leq 1$. Moreover, $\rho_{\rm{max}}(X, Y) = 0$ if and only if $X$ and $Y$ are independent, and $\rho_{\rm{max}}(X, Y) = 1$ if there exist functions $g, h$ such that $g(X) = h(Y)$ almost surely. For $X$ a binary random variable (see e.g.,~\cite{witsenhausen_1975}):
\begin{IEEEeqnarray}{rCl}
\rho^{2}_{\rm{max}}(X, Y) & = & \left[ \sum_{x \in \{0,1\}, y} \frac{(P_{X,Y}(x,y))^{2}}{P_X(x) P_Y(y)} \right] - 1. \nonumber
\end{IEEEeqnarray}

Given a $\BIC$ $(\sY, W)$, consider a random couple $(X,Y)$ with $X \sim \Bernoulli(1/2)$ and $P_{Y|X} = W$. In this case, we adopt the abbreviated notation $\rho^{2}_{\rm{max}}(W) := \rho^{2}_{\rm{max}}(X, Y)$. This is the squared maximal correlation parameter of $W$, assuming a uniform input distribution. If we take
\begin{align}
f(s) & = (2s - 1)^{2}, \notag
\end{align}
then evidently $\I_f(W) = \rho^{2}_{\rm{max}}(W)$. The following equalities establish this claim, with $P_X(x) = \frac{1}{2}$:
\begin{IEEEeqnarray}{rRl}
\IEEEeqnarraymulticol{3}{l} {\rho^{2}_{\rm{max}}(W) } \nonumber \\
~ & := & \left[ \sum_{x \in \{0,1\}, y} \frac{(P_{X,Y}(x,y))^{2}}{P_X(x) P_Y(y)} \right] - 1 \nonumber \\
~ & = & \left[ \sum_{y} 2P_Y(y)\left[ (P_{X|Y}(0|y))^{2} + (P_{X|Y}(1|y))^{2} \right] \right] - 1 \nonumber \\
~ & = & \E\left[2(\BlkFn_{W}(Y))^{2} + 2(1 - \BlkFn_{W}(Y))^{2} - 1\right] \nonumber \\
~ & = & \E\left[2S^{2} + 2(1 - S)^{2} - 1\right] \nonumber \\
~ & = & \E\left[ (2S - 1)^{2} \right] \nonumber \\
~ & \equiv & \I_f(W). \nonumber
\end{IEEEeqnarray}

\begin{remark}
{\em Consider functions $f: [0,1] \rightarrow \Reals$, $\varphi: \Reals \rightarrow \Reals$, and $\varphi \circ f: [0,1] \rightarrow \Reals$. Then $\I_f(W) = \E [f(S)]$ and $\I_{\varphi \circ f}(W) = \E [\varphi(f(S))]$. We note that $\I_{\varphi \circ f}(W) \neq \varphi(\I_f(W))$ in general. For example, if $\varphi(x) = x^{2}$, then clearly $\E [(f(S))^{2}] \neq \E [f(S)] \E[f(S)]$ in general. Examining the last two rows of Table~\ref{tbl:FunctionalsOfBICs}, $(1-2P_{\rm{e,ML}}(W))^{2} \neq \rho^{2}_{\rm{max}}(W)$ for arbitrary $\BIC$s, even though this is true for the special case that $W \equiv \BSC(p)$. Overall, the maximal correlation is a non-trivial, distinct parameter.} \end{remark}

\section{Output-symmetric $\BIC$s}\label{sec:OutputSymmetricBICs}

In his original pioneering work~\cite{arikan09}, \Arikan\ analyzed $\BIC$s having the property of output symmetry:

\begin{definition}[Output-symmetric $\BIC$]\label{def:OutputSymmetricBIC}
A $\BIC$ $(\sY,W)$ is output-symmetric if there exists a bijection $\pi : \sY \to \sY$, such that $\pi^{-1}=\pi$ and $W(\pi (y) |0)=W(y|1)$ for all $y \in \sY$.
\end{definition}
\noindent The phrases ``symmetric $\BIC$" and ``output-symmetric $\BIC$" will be used interchangeably.

\subsection{Structural decomposition of symmetric $\BIC$s}

Let us first recall the following definition:
\begin{definition}[Compound channel]\label{def:CompoundChannel}
Let $(\sY_i,W_i)$, $i \in [m] \deq \{1,\ldots,m\}$, be a collection of $\BIC$s, and let $\lambda = (\lambda_1,\ldots,\lambda_m)$ be a probability distribution on $[m]$. A {\em compound channel} with subchannels $\{W_i\}$ and mixing distribution $\lambda$ is a $\BIC$ $W$ defined by transition probabilities
$$
W(i,y|x) = \lambda_i W_i(y | x), \notag
$$
for all $x \in \{0,1\}$, $i \in [m]$, and $y \in \sY_i$. A compound channel will be denoted as $W = \bigoplus^m_{i=1} \lambda_i W_i$.
\end{definition}

\noindent The following structural result proved in~\cite[Theorem~2.1]{huber_2006} establishes that any symmetric $\BIC$ is a compound channel which consists of elementary subchannels that are $\BSC$s:
\begin{theorem}[Channel decomposition~\cite{huber_2006}]\label{thm:BIOSdecomp} Let the $\oplus$-operator applied to channels in Definition~\ref{def:CompoundChannel} define a compound channel. For any symmetric $\BIC$ $W$, there exist a positive integer $m$, a probability vector $\lambda = (\lambda_1,\ldots,\lambda_m)$, and error parameters $\GAM_1,\ldots,\GAM_m \in [0,1]$, such that
	\begin{align}\label{eq:BIOSdecomp}
	W \equiv  \bigoplus^m_{i=1} \lambda_i \BSC(\GAM_i).
	\end{align}
\end{theorem}

\subsection{Blackwell measures of symmetric $\BIC$s}\label{subsec:BlackwellMeasuresSymmetricBICs}

It is not difficult to show that the Blackwell measure of a compound channel is given by the mixture of the Blackwell measures of the constituent subchannels:
$$
\sm_W = \sum^m_{i=1}\lambda_i \sm_{W_i}.
$$
Thus, Theorem~\ref{thm:BIOSdecomp} shows that the Blackwell measure of any symmetric $\BIC$ is a mixture of Blackwell measures of BSCs. In particular, if $W \equiv \bigoplus^m_{i=1} \lambda_i \BSC(p_i)$, then
\begin{align*}
\sm_W &= \sum^m_{i=1} \lambda_i \sm_{\BSC(p_i)} \\
&= \sum^m_{i=1} \left( \frac{\lambda_i}{2}\delta_{p_i} + \frac{\lambda_i}{2}\delta_{\bar{p_i}}\right).
\end{align*}
Thus, any symmetric $\BIC$ $W$ that admits the decomposition \eqref{eq:BIOSdecomp} is  specified, up to Blackwell equivalence, by the set
\begin{align}\label{eq:Blackwell_set}
\cC_W \deq \left\{ (\lambda_i,p_i ) : i \in [m]\right\}.
\end{align}
Moreover, if $S \sim \sm_W$, then $\bar{S} = 1-S \sim \sm_W$ as well.

For two symmetric $\BIC$s $W$ and $W'$, the Blackwell ordering is equivalent to the \textit{symmetric convex ordering} introduced by Alsan \cite[Ch.~6]{alsan_phd14}, according to which $W$ dominates $W'$ if and only if
$$
\E[f(\Delta_W(Y))] \ge \E[f(\Delta_{W'}(Y'))]
$$
for all convex and even functions $f : [-1,1] \to \Reals$, where $(X, Y)$ is a random couple, $P_X = \Bernoulli(1/2)$, $P_{Y|X} = W$, and $P_{Y'|X} = W'$. The function $\Delta_W(\cdot)$ was defined in Eqn.~\eqref{eqn:DdistributionDef}.


\subsection{Examples of properties of symmetric $\BIC$s}

Theorem~\ref{thm:BIOSdecomp} also has implications for the computation of functionals of channels. Indeed, it follows directly from Eqn.~\eqref{eqn:BSCDeltaMeasure} and Eqn.~\eqref{eqn:FunctionalBICEqnDefinition} that, for any $f : [0,1] \to \Reals$,
$$
\I_f(\BSC(\GAM)) = \frac{1}{2}f(\GAM) + \frac{1}{2}f(\bar{\GAM}).
$$
Then, for a symmetric $\BIC$ $W$ with decomposition \eqref{eq:BIOSdecomp},
\begin{align}
\I_f(W) & = \sum^m_{i=1} \lambda_i \I_f(\BSC(\GAM_i)) \label{eqn:InducedFunctionalSymmBIC1} \\
& = \sum^m_{i=1} \frac{\lambda_i}{2} \left(f(\GAM_i) + f(\bar{\GAM}_i)\right). \label{eqn:InducedFunctionalCompForSymmetricBIC}
\end{align}
Consequently, for any symmetric $\BIC$ $(\sY, W)$, an induced functional of the form $\I_f(W)$ given in Table~\ref{tbl:FunctionalsOfBICs}, such as $I(W)$, $Z(W)$, $B_{\lambda}(W)$, $\rho^{2}_{\rm{max}}(W)$, etc., may be computed in terms of the functionals of elementary $\BSC$ subchannels. The channel dispersion $V(W)$ is not an induced functional, but may be derived from the induced functionals $M_2(W)$ and $I(W)$.

\begin{example}[Channel dispersion of $\BSC(\GAM)$ and $\BEC(\varepsilon)$]\label{example:DispersionBSCandBEC} {\em The channel dispersions of the $\BSC(\GAM)$ for $\GAM \notin \{0, \frac{1}{2}, 1\}$, and of the $\BEC(\varepsilon)$ are given by (cf.\cite[Theorem~52 and Theorem~53]{yury_finite_length_2010}):
\begin{align}
V(\BSC(\GAM)) & = \GAM \bar{\GAM} \left( \log_2 \frac{ \bar{\GAM} } { \GAM } \right)^{2}, \label{eqn:DispersionBSC} \\
V(\BEC(\varepsilon)) & = \varepsilon \bar{\varepsilon}. \label{eqn:DispersionBEC}
\end{align}
For $\GAM \in \{0, \frac{1}{2}, 1\}$, $V(\BSC(\GAM))$ approaches the limit of $0$. Eqn.~\eqref{eqn:DispersionBSC} can be derived from Eqn.~\eqref{eqn:InducedFunctionalCompForSymmetricBIC}. The second moment $M_2(\BSC(\GAM))$ is computed by selecting $f(s) = \psi_2(s)$:
\begin{align}
M_2(\BSC(\GAM)) & = \frac{ f(\GAM)}{2} + \frac{ f(\bar{\GAM})}{2} \notag \\
& = \GAM(1 + \log_2 \GAM)^{2} + \bar{\GAM}(1 + \log_2 \bar{\GAM})^{2}. \label{eqn:M2MomentForBSC}
\end{align}
The dispersion parameter is computed as
\begin{IEEEeqnarray}{rRl}
V(\BSC(\GAM)) & \deq & M_2(\BSC(\GAM)) - (I(\BSC(\GAM)))^{2} \notag \\
& = & \GAM \bar{\GAM} \left[ (\log_2 \GAM)^{2} + (\log_2 \bar{\GAM})^{2} - 2 \log_2 \GAM \log_2 \bar{\GAM}\right], \notag
\end{IEEEeqnarray}
which is verified to be equivalent to Eqn.~\eqref{eqn:DispersionBSC}. By observing that $\BEC(\varepsilon) \equiv  \bar{\varepsilon}\, \BSC(0) \oplus \varepsilon\, \BSC(1/2)$, it is straightforward to verify Eqn.~\eqref{eqn:DispersionBEC}.}
\end{example}
\noindent The calculation of a variety of channel parameters, e.g., the channel dispersion, for symmetric $\BIC$s is greatly simplified due to channel decompositions:

\begin{lemma}[Channel dispersion of an arbitrary symmetric $\BIC$]\label{lemma:DispersionSymmBIC}
Consider a symmetric $\BIC$ $(\sY, W)$ with decomposition $W \equiv  \bigoplus^m_{i=1} \lambda_i \BSC(\GAM_i)$. The channel capacity $I(W)$ and channel dispersion $V(W)$ may be written in terms of the capacities and dispersions of the subchannels:
\begin{IEEEeqnarray}{rCll}
I(W) & = & & \sum_{i=1}^{m} \lambda_i I(\BSC(\GAM_i)), \label{eqn:EqnIWSymmBIC} \\
V(W) & = & \Biggl( & \sum_{i=1}^{m} \lambda_i V(\BSC(\GAM_i)) \nonumber \\
& & +\> & \sum_{i=1}^{m} \lambda_i \left( I(\BSC(\GAM_i)) - I(W) \right)^{2} \Biggl). \label{eqn:EqnVWSymmBIC}
\end{IEEEeqnarray}
\end{lemma}
\begin{proof} Omitted. 
\end{proof}
\begin{remark}{\em Eqn.~\eqref{eqn:EqnIWSymmBIC} is also discussed in~\cite[Eqn.~2.7]{huber_2006}, which cites Gallager's approach for computing the capacity of symmetric channels~\cite[Ch.~4]{gallager68}. In a similar manner, it is straightforward to derive Eqn.~\eqref{eqn:EqnVWSymmBIC} from first principles by exploiting channel decompositions.
}
\end{remark}


\subsection{Mutual information profile}

The mutual information profile ($\MIP$) (see \cite[Ch.~2]{huber_2006} for a detailed presentation) is based on the structural decomposition of symmetric $\BIC$s:

\begin{definition}[Mutual Information Profile]\label{def:MIP} A symmetric $\BIC$ $(\sY, W)$ with structural decomposition $W \equiv  \bigoplus^m_{i=1} \lambda_i \BSC(\GAM_i)$ as in Theorem~\ref{thm:BIOSdecomp} is uniquely characterized by a random variable $\Phi$ that takes values in the unit interval $[0,1]$ according to the probability law
\begin{align}
\sm^\Phi_{W} & \deq \sum_{i=1}^{m} \lambda_i \delta_{I(\BSC(\GAM_i))}. \label{eqn:ProbMeasureOfMIP}
\end{align}
The probability law $\sm^\Phi_W$ is called the {\em mutual information profile (MIP)} of the channel $W$.
\end{definition}
Similar to the Blackwell measure which uniquely specifies an arbitrary $\BIC$ up to Blackwell equivalence, the $\MIP$ uniquely specifies $\BIC$s with the property of output symmetry. In fact, it is easy to see from Eqn~\eqref{eqn:ProbMeasureOfMIP} that the MIP $\sm^\Phi_W$ is simply the probability law of $1-h_2(S)$ when $S \sim \sm_W$.

\section{The polar transform}\label{sec:PolarTransformOfChannels}

The polar transform maps a pair of $\BIC$s to another pair of transformed $\BIC$s via a Boolean {\sc xor} of the binary inputs of the original channels~\cite{arikan09}. The Boolean {\sc xor} creates dependence between the random variables associated to the inputs and outputs of the transformed $\BIC$s.

\subsection{The polar transform}

\begin{definition}[The polar transform~\cite{arikan09}]\label{def:PolarTransform}
The polar transform maps a pair of $\BIC$s $(\sY_1,W_1)$ and $(\sY_2,W_2)$ into another pair of $\BIC$s $(\sY_1 \times \sY_2, W_1 \ATbad W_2)$ and $(\sY_1 \times \sY_2 \times \{0,1\},W_1 \ATgood W_2)$ as follows:
\begin{subequations}\label{eq:polar_transform}
\begin{align}
	&(W_1 \ATbad W_2) (y_1,y_2|x) \nonumber\\
	& \qquad \qquad ~ \deq \frac{1}{2}\sum_{u \in \{0,1\}}W_1(y_1 | u \oplus x) W_2(y_2|u) \\
	&(W_1 \ATgood W_2) (y_1,y_2,u|x) \nonumber\\
	& \qquad \qquad ~ \deq \frac{1}{2} W_1(y_1|u \oplus x) W_2(y_2|x)
\end{align}
\end{subequations}
for all $x, u \in \{0,1\}$ and all $(y_1,y_2) \in \sY_1 \times \sY_2$, where $\oplus$ is the Boolean {\sc xor}.
\end{definition}

\noindent The transformed channel $W_1 \ATbad W_2$ is ``weaker" than both $W_1$ and $W_2$ as will be clarified in subsequent analysis. The transformed channel $W_1 \ATgood W_2$ is improved because it is equivalent to decoding based on two independent noisy versions of the binary input. A parallel broadcast of the binary input is formalized as follows:

\begin{definition}[Product $\BIC$ $W_1 \times W_2$]
Given two $\BIC$s $(\sY_1,W_1)$ and $(\sY_2,W_2)$, we define the product $\BIC$ $(\sY_1 \times \sY_2,W_1 \times W_2)$ by
$$
(W_1 \times W_2)(y_1,y_2|x) \deq W_1(y_1|x) W_2(y_2|x)
$$
for all $x \in\{0,1\}$ and all $(y_1,y_2) \in \sY_1 \times \sY_2$. In other words, $W_1 \times W_2$ is the parallel broadcast channel formed by $W_1$ and $W_2$ \cite{ziv_shamai_2005}.
\end{definition}

\subsection{Blackwell measures of one-step polarized $\BIC$s}

Consider two Blackwell measures $\sm_1,\sm_2  \in \cM$. The following operations $\ATbad$ and $\ATgood$ on a pair of Blackwell measures yield two additional probability measures $\sm_1 \ATbad \sm_2$ and $\sm_1 \ATgood \sm_2$ on $[0,1]$.

\begin{definition}\label{def:ATbadATgoodProbMeasures}
Let $\sm_1,\sm_2  \in \cM$ where $\cM$ is the space of probability measures as defined in Theorem~\ref{thm:BSS}. Let $S_1 \sim \sm_1$ and $S_2 \sim \sm_2$ be two independent random variables. The probability measures $\sm_1 \ATbad \sm_2$ and $\sm_1 \ATgood \sm_2$ are defined as follows: For any continuous bounded $f : [0,1] \to \Reals$, let
\begin{align}\label{eq:bad_BM}
	\int_{[0,1]}f\d(\sm_1 \ATbad \sm_2) = \E[f(1-S_1 \star S_2)]
\end{align}
and
\begin{align}\label{eq:good_BM}
	\int_{[0,1]}f\d(\sm_1 \ATgood \sm_2) &= \E\Bigg[(1-S_1 \star S_2)f\left(\frac{S_1S_2}{1-S_1 \star S_2}\right) \nonumber\\
	& \qquad + (S_1 \star S_2) f\left(\frac{\bar{S}_1 S_2}{S_1 \star S_2}\right) \Bigg].
\end{align}
\end{definition}

\begin{lemma} The probability measures $\sm_1 \ATbad \sm_2$ and $\sm_1 \ATgood \sm_2$ are also Blackwell measures.
\end{lemma}
\begin{proof} Setting $f(s) = s$ in Eqs.~\eqref{eq:bad_BM} and \eqref{eq:good_BM}, and recalling that $S_1$ and $S_2$ are independent and both have mean $\frac{1}{2}$, we get $\int_{[0,1]} s(\sm_1 \ATbad \sm_2)(\d s) = \E[1-S_1 \star S_2] = \frac{1}{2}$. Similarly, $\int_{[0,1]} s(\sm_1 \ATgood \sm_2)(\d s) = \E[S_2] = \frac{1}{2}$. Thus, both $\sm_1 \ATbad \sm_2$ and $\sm_1 \ATgood \sm_2$ are in $\cM$.
\end{proof}

\noindent Regarding the following theorem, one can also refer to density evolution for polar codes~\cite{polar_code_density_evolution_2009}. The Blackwell measures of one-step polarized channels $W_1 \ATbad W_2$ and $W_1 \ATgood W_2$ can be written in terms of the Blackwell measures of $W_1$ and $W_2$:
\begin{theorem}[Evolution of Blackwell measures under \Arikan's polar transform]\label{thm:BM_polarization} The Blackwell measures of the one-step polarized $\BIC$s $W_1 \ATbad W_2$ and $W_1 \ATgood W_2$ introduced in Definition~\ref{def:PolarTransform} are given by
\begin{align}
\sm_{W_1 \ATbad W_2}  & = \sm_{W_1} \ATbad \sm_{W_2}, \notag \\
\sm_{W_1 \ATgood W_2} & = \sm_{W_1} \ATgood \sm_{W_2}, \notag
\end{align}
where the operations $\ATbad$ and $\ATgood$ on Blackwell measures were defined in Definition~\ref{def:ATbadATgoodProbMeasures}.
\end{theorem}
\begin{proof} We first establish the formula for $W_1 \ATbad W_2$. Let $(X_i,Y_i)$, for $i \in \{1,2\}$, where $(X_1,Y_1)$ and $(X_2,Y_2)$ are independent, $P_{X_1} = P_{X_2} = \Bernoulli(1/2)$, and $P_{Y_i|X_i} = W_i$. Then, recalling the definition of $\BlkFn_{W}$ in \eqref{eqn:BlackwellMeasureBIC}, and using an abbreviated notation $\bar{\BlkFn}_{W} = 1 - \BlkFn_{W}$, we can write
\begin{IEEEeqnarray*}{lCl}
\IEEEeqnarraymulticol{3}{l} { (W_1 \ATbad W_2)(y_1,y_2|0)  } \\
 & = & \frac{1}{2} \left(W_1(y_1|0)W_2(y_2|0)  + W_1(y_1|1)W_2(y_2|1)\right) \\
 & = & 2P_{Y_1}(y_1)P_{Y_2}(y_2) \left( \BlkFn_{W_1}(y_1)\BlkFn_{W_2}(y_2) + \bar{\BlkFn}_{W_1}(y_1)\bar{\BlkFn}_{W_2}(y_2)\right) \\
 & = & 2P_{Y_1}(y_1)P_{Y_2}(y_2) \cdot \left(1-\BlkFn_{W_1}(y_1) \star \BlkFn_{W_2}(y_2)\right)
\end{IEEEeqnarray*}
and
\begin{IEEEeqnarray*}{lCl}
\IEEEeqnarraymulticol{3}{l} { (W_1 \ATbad W_2)(y_1,y_2|1) }  \\
 & = & \frac{1}{2} \left(W_1(y_1|1)W_2(y_2|0)  + W_1(y_1|0)W_2(y_2|1)\right) \\
 & = & 2P_{Y_1}(y_1)P_{Y_2}(y_2) \left( \bar{\BlkFn}_{W_1}(y_1)\BlkFn_{W_2}(y_2) + \BlkFn_{W_1}(y_1)\bar{\BlkFn}_{W_2}(y_2)\right) \\
 & = & 2P_{Y_1}(y_1)P_{Y_2}(y_2) \cdot \left(\BlkFn_{W_1}(y_1) \star \BlkFn_{W_2}(y_2)\right).
\end{IEEEeqnarray*}
Thus, combining the above two equations,
\begin{IEEEeqnarray*}{rCl}
(W_1 \ATbad W_2)(y_1,y_2|0) + (W_1 \ATbad W_2)(y_1,y_2|1) & & \nonumber \\
\quad \quad \quad = 2P_{Y_1}(y_1)P_{Y_2}(y_2). & & \nonumber
\end{IEEEeqnarray*}
This yields
\begin{IEEEeqnarray*}{rCl}
\IEEEeqnarraymulticol{3}{l} {  \BlkFn_{W_1 \ATbad W_2}(y_1,y_2)  } \\
& \quad = & \frac{(W_1 \ATbad W_2)(y_1,y_2|0)}{(W_1 \ATbad W_2)(y_1,y_2|0) + (W_1 \ATbad W_2)(y_1,y_2|1)} \\
& \quad = & 1 - \BlkFn_{W_1}(y_1) \star \BlkFn_{W_2}(y_2).
\end{IEEEeqnarray*}
This shows that $S = \BlkFn_{W_1 \ATbad W_2}(Y_1,Y_2) = 1 - S_1 \star S_2$, where
 $S_1 = \BlkFn_{W_1}(Y_1)$ and $S_2 = \BlkFn_{W_2}(Y_2)$ are independent. Thus, for any continuous $f : [0,1] \to \Reals$,
\begin{align}
\int_{[0,1]} f \d\sm_{W_1 \ATbad W_2}(\d s) &= \E[f(S)] \\
&= \E[f(1-S_1 \star S_2)] \\
&= \int_{[0,1]} f \d(\sm_{W_1} \ATbad \sm_{W_2}).
\end{align}
We turn to $W_1 \ATgood W_2$. Let the $\oplus$-operator applied to channels in Definition~\ref{def:CompoundChannel} define a compound channel. From the definition of the polar transform given in Eqn.~\eqref{eq:polar_transform}, it follows that
\begin{align}\label{eq:good_decomp}
	W_1 \ATgood W_2 \equiv \frac{1}{2}W^{(0)} \oplus \frac{1}{2}W^{(1)}
\end{align}
with $W^{(0)} \deq W_1 \times W_2$ and $W^{(1)} \deq \bar{W}_1 \times W_2$, where $\bar{W}_1$ is the $\BIC$ related to $W_1$ via $\bar{W}_1(\cdot|x) = W_1(\cdot|\bar{x})$. Then the random variables $S^{(0)} \sim \sm_{W^{(0)}}$ and $S^{(1)} \sim \sm_{W^{(1)}}$ evidently satisfy
$$
\E[f(S^{(0)})] = 2\,\E\left[(1-S_1 \star S_2)f\left(\frac{S_1 S_2}{1-S_1 \star S_2}\right)\right]
$$
and
$$
\E[f(S^{(1)})] = 2\,\E\left[(S_1 \star S_2)f\left(\frac{\bar{S}_1 S_2}{S_1 \star S_2}\right)\right]
$$
for every continuous $f : [0,1] \to \Reals$. Combining this result with \eqref{eq:good_decomp} yields:

\begin{align*}
\int_{[0,1]} f\d\sm_{W_1 \ATgood W_2} &= \frac{1}{2}\E[f(S^{(0)})] + \frac{1}{2}\E[f(S^{(1)})] \\
	&= \E\Bigg[(1-S_1 \star S_2)f\left(\frac{S_1 S_2}{1-S_1 \star S_2}\right)\\
	& \qquad \qquad +(S_1 \star S_2)f\left(\frac{\bar{S}_1 S_2}{S_1 \star S_2}\right)\Bigg]\\
	&= \int_{[0,1]} f\d(\sm_{W_1} \ATgood \sm_{W_2}).
\end{align*}
Since $f$ is arbitrary, we obtain the formula for $\sm_{W_1 \ATgood W_2}$.
\end{proof}

\begin{figure*}[t]
\begin{center}


\scalebox{1} 
{

\begin{pspicture}(0,-6.0)(11.0,3.5)

\definecolor{color384}{rgb}{0.8,0.8,0.8}

%
%
\rput(-1.0, -0.5){
\rput(-1.5,-0.5){
\psline[linewidth=0.06cm,linecolor=color384](0.2,0.68)(0.2,3.7)
\psline[linewidth=0.06cm,linecolor=color384](0.02,1.2)(3.2,1.2)
\usefont{T1}{ptm}{m}{n}
\rput(2.2,0.938){$1$}
\psline[linewidth=0.06cm,arrowsize=0.15cm,arrowlength=1.2,arrowinset=0.01]{C->}(2.2,1.2)(2.2,2.80)
\psline[linewidth=0.06cm,arrowsize=0.15cm,arrowlength=1.2,arrowinset=0.01]{C->}(0.2,1.2)(0.2,1.6)
\usefont{T1}{ptm}{m}{n}
\rput(2.2,3.0){$1-\epsilon$}
\usefont{T1}{ptm}{m}{n}
\rput(0.4,1.8){$\epsilon$}
\rput(1.2, 0.25){$\sm_{\Phi, \BEC(\epsilon)}$}
}

\rput(2.5,-0.5){
\psline[linewidth=0.06cm,linecolor=color384](0.2,0.68)(0.2,3.7)
\psline[linewidth=0.06cm,linecolor=color384](0.02,1.2)(3.2,1.2)
\usefont{T1}{ptm}{m}{n}
\rput(2.2,0.938){$1$}
\psline[linewidth=0.06cm,arrowsize=0.15cm,arrowlength=1.2,arrowinset=0.01]{C->}(2.2,1.2)(2.2,2.6)
\psline[linewidth=0.06cm,arrowsize=0.15cm,arrowlength=1.2,arrowinset=0.01]{C->}(0.2,1.2)(0.2,1.8)
\usefont{T1}{ptm}{m}{n}
\rput(2.2,2.8){$1-\tau$}
\usefont{T1}{ptm}{m}{n}
\rput(0.4,2.0){$\tau$}
\rput(1.2, 0.25){$\sm_{\Phi, \BEC(\tau)}$}
}

\rput(6.5,-0.5){
\psline[linewidth=0.06cm,linecolor=color384](0.2,0.68)(0.2,3.7)
\psline[linewidth=0.06cm,linecolor=color384](0.02,1.2)(3.2,1.2)
\usefont{T1}{ptm}{m}{n}
\rput(2.2,0.938){$1$}
\psline[linewidth=0.06cm,arrowsize=0.15cm,arrowlength=1.2,arrowinset=0.01]{C->}(2.2,1.2)(2.2,2.32)
\psline[linewidth=0.06cm,arrowsize=0.15cm,arrowlength=1.2,arrowinset=0.01]{C->}(0.2,1.2)(0.2,2.08)
\usefont{T1}{ptm}{m}{n}
\rput(2.2,2.52){$\bar{\epsilon}\bar{\tau}$}
\usefont{T1}{ptm}{m}{n}
\rput(0.75,2.28){$1-\bar{\epsilon}\bar{\tau}$}
\rput(1.75, 0.25){$\sm_{\Phi, \BEC(\epsilon) \ATbad \BEC(\tau)}$}
}

\rput(10.5,-0.5){
\psline[linewidth=0.06cm,linecolor=color384](0.2,0.68)(0.2,3.7)
\psline[linewidth=0.06cm,linecolor=color384](0.02,1.2)(3.2,1.2)
\usefont{T1}{ptm}{m}{n}
\rput(2.2,0.938){$1$}
\psline[linewidth=0.06cm,arrowsize=0.15cm,arrowlength=1.2,arrowinset=0.01]{C->}(2.2,1.2)(2.2,3.08)
\psline[linewidth=0.06cm,arrowsize=0.15cm,arrowlength=1.2,arrowinset=0.01]{C->}(0.2,1.2)(0.2,1.45)
\usefont{T1}{ptm}{m}{n}
\rput(2.2,3.28){$1-\epsilon \tau$}
\usefont{T1}{ptm}{m}{n}
\rput(0.47,1.6){$\epsilon \tau$}
\rput(1.75, 0.25){$\sm_{\Phi, \BEC(\epsilon) \ATgood \BEC(\tau)}$}
}
}

%
%
\rput(-1.0, -5.5){
\rput(-1.5,-0.5){
\psline[linewidth=0.06cm,linecolor=color384](0.2,0.68)(0.2,3.7)
\psline[linewidth=0.06cm,linecolor=color384](0.02,1.2)(3.2,1.2)
\usefont{T1}{ptm}{m}{n}
\rput(1.262, 0.938){$~_{I(\BSC(\GAM))}$}
\psline[linewidth=0.06cm,arrowsize=0.15cm,arrowlength=1.2,arrowinset=0.01]{C->}(1.262,1.2)(1.262,3.2)
\usefont{T1}{ptm}{m}{n}
\rput(1.262,3.4){$1$}
\rput(1.2, 0.25){$\sm_{\Phi, \BSC(\GAM)}$}
}

\rput(2.5,-0.5){
\psline[linewidth=0.06cm,linecolor=color384](0.2,0.68)(0.2,3.7)
\psline[linewidth=0.06cm,linecolor=color384](0.02,1.2)(3.2,1.2)
\usefont{T1}{ptm}{m}{n}
\rput(0.9803, 0.938){$~_{I(\BSC(\NU))}$}
\psline[linewidth=0.06cm,arrowsize=0.15cm,arrowlength=1.2,arrowinset=0.01]{C->}(0.9803,1.2)(0.9803,3.2)
\usefont{T1}{ptm}{m}{n}
\rput(0.9803,3.4){$1$}
\rput(1.2, 0.25){$\sm_{\Phi, \BSC(\NU)}$}
}

\rput(6.5,-0.5){
\psline[linewidth=0.06cm,linecolor=color384](0.2,0.68)(0.2,3.7)
\psline[linewidth=0.06cm,linecolor=color384](0.02,1.2)(3.2,1.2)
\usefont{T1}{ptm}{m}{n}
\rput(1.0497, 0.938){$~_{I(\BSC(\GAM \star \NU))}$}
\psline[linewidth=0.06cm,arrowsize=0.15cm,arrowlength=1.2,arrowinset=0.01]{C->}(0.6797,1.2)(0.6797,3.2)
\usefont{T1}{ptm}{m}{n}
\rput(0.6797,3.4){$1$}
\rput(1.75, 0.25){$\sm_{\Phi, \BSC(\GAM) \ATbad \BSC(\NU)}$}
}

\rput(10.5,-0.5){
\psline[linewidth=0.06cm,linecolor=color384](0.2,0.68)(0.2,3.7)
\psline[linewidth=0.06cm,linecolor=color384](0.02,1.2)(3.2,1.2)
\usefont{T1}{ptm}{m}{n}
\rput(0.872,0.938){$~_{I(\BSC(\beta))}$}
\rput(1.9258,0.638){$_{I(\BSC(\alpha))}$}
\psline[linewidth=0.06cm,arrowsize=0.15cm,arrowlength=1.2,arrowinset=0.01]{C->}(1.9258,1.2)(1.9258,2.86)
\psline[linewidth=0.06cm,arrowsize=0.15cm,arrowlength=1.2,arrowinset=0.01]{C->}(0.4252,1.2)(0.4252,1.54)
\usefont{T1}{ptm}{m}{n}
\rput(1.9258,3.06){$1-\GAM \star \NU$}
\usefont{T1}{ptm}{m}{n}
\rput(0.712,1.74){$\GAM \star \NU$}
\rput(1.75, 0.25){$\sm_{\Phi, \BSC(\GAM) \ATgood \BSC(\NU)}$}
}
}

\end{pspicture}
}

\end{center}\small
\caption{The mutual information profiles of one-step polarized channels as a weighted sum of Dirac measures (see Definition~\ref{def:MIP}). The vertical axis depicts weights in the interval $[0,1]$ assigned to the Dirac measures. The horizontal axis depicts the location in the interval $[0,1]$ of the Dirac measures in the profile.} \label{fig:PolarizationOneStep_BEC_BSC}
\end{figure*}

\subsection{Blackwell measures of one-step polarized BSCs}

Although Theorem~\ref{thm:BM_polarization} fully characterizes the Blackwell measure of one-step transformed $\BIC$s, further specialization is possible if the original $\BIC$s have the property of symmetry. Symmetric $\BIC$s are composed of elementary subchannels (i.e., $\BSC$s) as established in Theorem~\ref{thm:BIOSdecomp}. Thus, the essential aspect of polarization for symmetric $\BIC$s is the interaction of subchannels that are $\BSC$s which may have \emph{different} probabilities of error. The following lemmas provide complementary descriptions of the effect of the polar transform on a pair of $\BSC$s:

\begin{lemma}[One-step polarization of $\BSC$s --- transition matrices]\label{lemma:TransitionMatBSC}
Let $(\GAM, \NU) \in [0,1] \times [0,1]$. Consider the channels $\BSC(\GAM)$ and $\BSC(\NU)$. As defined in Definition~\ref{def:PolarTransform}, let $\BSC(\GAM) \ATbad \BSC(\NU)$ and $\BSC(\GAM) \ATgood \BSC(\NU)$ represent the transformed channels. The corresponding transition matrices have the following structure:
\begin{align}
\tilde{T}_{\BSC(\GAM) \ATbad \BSC(\NU)} & = T_{\BSC(\GAM \star \NU)} \notag \\
& = \left[ \begin{array}{cc} 1-\GAM \star \NU & \GAM \star \NU \\ \GAM \star \NU & 1 - \GAM \star \NU \end{array} \right], \label{eqn:TransitionMatATbad} \\
\tilde{T}_{\BSC(\GAM) \ATgood \BSC(\NU)} & = T_{\BSC(\GAM) \times \BSC(\NU)} \notag \\
& = \left[ \! \! \begin{array}{cccc} ~\bar{\GAM} \bar{\NU} & ~\GAM \NU & ~\GAM \bar{\NU} & ~\bar{\GAM} \NU \\ ~\GAM \NU & ~\bar{\GAM}\bar{\NU} & ~\bar{\GAM} \NU & ~\GAM \bar{\NU} \end{array}\right], \label{eqn:TransitionMatATgood}
\end{align}
where $\tilde{T}_{W}$ represents the transition matrix $T_{W}$ with a reduced number of columns due to aggregating (i.e., grouping) output symbols of $\BIC$ $W$.
\end{lemma}
\begin{proof} Provided in Appendix~\ref{appendix:TransitionMatPolarizedBSC}.
\end{proof}
\begin{lemma}[One-step polarization of $\BSC$s --- structural decomposition]\label{lemma:StructuralPolarizedBSC}
Let $(\GAM, \NU) \in [0,\frac{1}{2}] \times [0,\frac{1}{2}]$. Assume $(\GAM, \NU) \neq (0,0)$ so that $(\GAM \star \NU) \neq 0$. Consider the channels $\BSC(\GAM)$ and $\BSC(\NU)$ as in Lemma~\ref{lemma:TransitionMatBSC}, and define the following parameters:
\begin{align}
\alpha & := \frac{ \GAM \NU }{ 1 - \GAM \star \NU },   \label{eqn:AlphaNewBSCErrorProb} \\
\beta  & := \frac{ \bar{\GAM} \NU }{ \GAM \star \NU }, \label{eqn:BetaNewBSCErrorProb}
\end{align}
where $\alpha \in [0,\frac{1}{2}]$ and $(\beta \wedge \bar{\beta}) \in [0,\frac{1}{2}]$. Then the one-step polarized channels satisfy the following equivalences.
\begin{align}
& \BSC(\GAM) \ATbad \BSC(\NU)  \equiv \BSC(\GAM \star \NU), \label{eqn:StructuralDecompPolarizedBSCATbad} \\
& \BSC(\GAM) \ATgood \BSC(\NU) \notag \\
& \quad \quad \quad \equiv \BSC(\GAM) \times \BSC(\NU) \notag \\
& \quad \quad \quad \equiv \left( 1 \! - \! \GAM \star \NU \right) \BSC(\alpha) \oplus \left( \GAM \star \NU \right) \BSC(\beta \wedge \bar{\beta}). \label{eqn:StructuralDecompPolarizedBSCATgood}
\end{align}
\end{lemma}
\begin{proof} Provided in Appendix~\ref{appendix:StructuralPolarizedBSC}.
\end{proof}

Lemmas~\ref{lemma:TransitionMatBSC} and \ref{lemma:StructuralPolarizedBSC} precisely characterize the interaction of a $\BSC(\GAM)$ with a $\BSC(\NU)$ after one step of polarization. The equivalences are based on equivalences between corresponding transition matrices. In a more general approach based on Blackwell measures, the equivalences may be derived directly from Theorem~\ref{thm:BM_polarization} as a corollary:

\begin{corollary}[One-step polarization of BSCs --- Blackwell measures]\label{corollary:BlackwellMeasuresPolarizedBSC} Let $(\GAM, \NU) \in [0,\frac{1}{2}] \times [0,\frac{1}{2}]$. Consider the channels $\BSC(\GAM)$ and $\BSC(\NU)$. Then the Blackwell measures of the transformed channels $\BSC(\GAM) \ATbad \BSC(\NU)$ and $\BSC(\GAM) \ATgood \BSC(\NU)$ are given as follows:
\begin{align}
\sm_{\BSC(\GAM) \ATbad \BSC(\NU)}  & = \sm_{\BSC(\GAM)} \ATbad \sm_{\BSC(\NU)} \notag \\
 & = \sm_{\BSC(\GAM \star \NU)}. \label{eqn:corollary:BlackwellMeasureBadPolarizedBSC} \\
\sm_{\BSC(\GAM) \ATgood \BSC(\NU)} & = \sm_{\BSC(\GAM)} \ATgood \sm_{\BSC(\NU)} \notag \\
 & = \sm_{\BSC(\GAM) \times \BSC(\NU)}.  \label{eqn:corollary:BlackwellMeasureGoodPolarizedBSC}
\end{align}
In addition, consider the parameters $\alpha := \frac{ \GAM \NU }{ 1 - \GAM \star \NU }$ and $\beta := \frac{ \bar{\GAM} \NU }{ \GAM \star \NU }$ as defined in Lemma~\ref{lemma:StructuralPolarizedBSC}. Then for $(\GAM, \NU) \neq (0,0)$, the parallel broadcast channel $\BSC(\GAM) \times \BSC(\NU)$ has the following representation in terms of Blackwell measure:
\begin{align}
& \sm_{\BSC(\GAM) \times \BSC(\NU)} \notag \\
& \quad \quad \quad =  \left( 1 - \GAM \star \NU \right) \sm_{\BSC(\alpha)}  + \left( \GAM \star \NU \right) \sm_{\BSC(\beta \wedge \bar{\beta})}. \label{eqn:corollary:BlackwellMeasureGoodParallelBroadcast}
\end{align}
\end{corollary}
\begin{proof} Provided in Appendix~\ref{appendix:BlackwellMeasuresPolarizedBSC}.
\end{proof}


\begin{remark} \em{ The lower-half panel of Figure~\ref{fig:PolarizationOneStep_BEC_BSC} depicts the polar transformation applied to two $\BSC$s. The channel $\BSC(\GAM) \ATbad \BSC(\NU)$ is a $\BSC$ with probability of error $p \star q$, which is larger than both $p$ and $q$. The channel $\BSC(\GAM) \ATgood \BSC(\NU)$ is a more complex channel. More precisely, it is a compound channel, as defined in Definition~\ref{def:CompoundChannel}, which consists of two elementary subchannels that are $\BSC$s. }
\end{remark}

\subsection{Blackwell measures of one-step polarized symmetric $\BIC$s}

Building upon Corollary~\ref{corollary:BlackwellMeasuresPolarizedBSC}, it is possible to characterize the image of a pair of arbitrary symmetric $\BIC$s under the polar transformation. The one-step polarization of symmetric $\BIC$s is characterized entirely by the interaction of elementary subchannels that are $\BSC$s.

\begin{corollary}[One-step polarization of symmetric $\BIC$s --- Blackwell measures]\label{corollary:BIOSPolarization} Consider two symmetric $\BIC$s $(\sY_1, W_1)$ and $(\sY_2, W_2)$. As established by Theorem~\ref{thm:BIOSdecomp}, there exist positive integers $m, k$, probability vectors $\lambda = (\lambda_1,\ldots,\lambda_m)$, $\mu = (\mu_1,\ldots,\mu_k)$, and parameters $p_1,\ldots,p_m \in [0,1]$, $q_1,\ldots,q_k \in [0,1]$ such that
\begin{align}
W_1 & \equiv  \bigoplus^m_{i=1} \lambda_i \BSC(\GAM_i), \notag \\
W_2 & \equiv  \bigoplus^k_{j=1} \mu_j \BSC(\NU_j). \notag
\end{align}
Then the transformed channels $W_1 \ATbad W_2$ and $W_1 \ATgood W_2$ have the following Blackwell measures which illustrate their structural decompositions:
\begin{align}
\sm_{W_1 \ATbad W_2} & = \sum^m_{i=1}\sum^k_{j=1}\lambda_i \mu_j \sm_{\BSC(\GAM_i \star \NU_j)}. \label{eq:BIOSbad} \\
\sm_{W_1 \ATgood W_2} & = \sum^m_{i=1}\sum^k_{j=1} \lambda_i \mu_j \sm_{\BSC(\GAM_i) \times \BSC(\NU_j)}. \label{eq:BIOSgood}
\end{align}
\end{corollary}
\begin{proof} Provided in Appendix~\ref{appendix:BlackwellMeasuresPolarizedSymmetricBICs}.
\end{proof}

\noindent As an illustration of Corollary~\ref{corollary:BIOSPolarization}, we give an alternative derivation of the fact that the image of a pair of $\BEC$s under the polar transform is another pair of $\BEC$s~\cite[Prop.~6]{arikan09}.

\begin{example}[One-step polarization of $\BEC$s --- Blackwell measures]\label{example:BECPolarizationBlackwellMeasures}
{\em Consider two erasure channels, $\BEC(\varepsilon)$ and $\BEC(\tau)$, with erasure probabilities $\varepsilon \in [0,1]$ and $\tau \in [0,1]$. The upper-half panel of Figure~\ref{fig:PolarizationOneStep_BEC_BSC} depicts that $\BEC(\varepsilon) \ATbad \BEC(\tau) \equiv \BEC(1 - \bar{\varepsilon}\bar{\tau})$ and $\BEC(\varepsilon) \ATgood \BEC(\tau) \equiv \BEC(\varepsilon \tau)$. We give a proof of the second equivalence here, and the proof of the first equivalence follows from nearly identical steps.
Applying Eqn.~\eqref{eq:BIOSgood} of Corollary~\ref{corollary:BIOSPolarization}, the channel $\BEC(\varepsilon) \ATgood \BEC(\tau)$ has the following Blackwell measure:
\begin{IEEEeqnarray*}{rCl}
\IEEEeqnarraymulticol{3}{l} { \sm_{\BEC(\varepsilon) \ATgood \BEC(\tau)} } \\
& = & \sm_{\BEC(\varepsilon)} \ATgood \sm_{\BEC(\tau)}  \\
& = & \left( \bar{\varepsilon} \sm_{\BSC(0)} + \varepsilon \sm_{\BSC(1/2)} \right) \ATgood \left( \bar{\tau} \sm_{\BSC(0)} + \tau \sm_{\BSC(1/2)} \right) \\
& = & \bar{\varepsilon}\bar{\tau} \sm_{\BSC(0) \times \BSC(0)} + \bar{\varepsilon}\tau \sm_{\BSC(0) \times \BSC(1/2)} \\
& & \> + \varepsilon \bar{\tau}\sm_{\BSC(1/2) \times \BSC(0)} + \varepsilon \tau\sm_{\BSC(1/2) \times \BSC(1/2)} \\
& = & \bar{\varepsilon}\bar{\tau} \sm_{\BSC(0)} + \bar{\varepsilon}\tau \sm_{\BSC(0)} + \varepsilon \bar{\tau}\sm_{\BSC(0)} + \varepsilon \tau\sm_{\BSC(1/2)} \\
& = & (1-\varepsilon\tau)\sm_{\BSC(0)} + \varepsilon \tau \sm_{\BSC(1/2)} \\
& = & \sm_{\BEC(\varepsilon \tau)}.
\end{IEEEeqnarray*}}
\end{example}


\section{The Construction of Polar Codes}\label{sec:SuccessiveChannelPolarization}

For a given $\BIC$ $(\sY, W)$, the polar transforms defined in Definition~\ref{def:PolarTransform} may be applied iteratively~\cite{arikan09}. Polarizing the original channel $W$ successively over $n$ iterations results in one of $2^{n}$ possible channels.

\begin{definition}[Successive channel polarization]\label{def:SuccessivePolarization} Consider a $\BIC$ $(\sY, W)$. Let $b = (b_0, b_1, b_2, \ldots, b_n)$ be a binary vector, $b \in \{0,1\}^{n+1}$. The channel $W_{b}$ is obtained by iterative polarization:
\begin{IEEEeqnarray}{rRl}
W_b & = & W_{(b_0,b_1,b_2,\ldots,b_n)} \notag \\
& := & \begin{cases}
W_{(b_0,b_1,b_2,\ldots,b_{n-1})} \ATbad  W_{(b_0,b_1,b_2,\ldots,b_{n-1})}, & \text{if $b_n = 0$,} \\
W_{(b_0,b_1,b_2,\ldots,b_{n-1})} \ATgood W_{(b_0,b_1,b_2,\ldots,b_{n-1})}, & \text{if $b_n = 1$.} \notag
\end{cases}
\end{IEEEeqnarray}
for all integers $n > 0$. The base case is given by $W_{(b_0)} \equiv W$.
\end{definition}

\subsection{Polar code construction: symmetric $\BIC$s}

\begin{figure}[t]
\begin{center}


\scalebox{0.92} 
{
\begin{pspicture}(0,-4.22)(5.2028127,4.22)

\rput(2.82, -0.25) {
\includegraphics[scale=0.442]{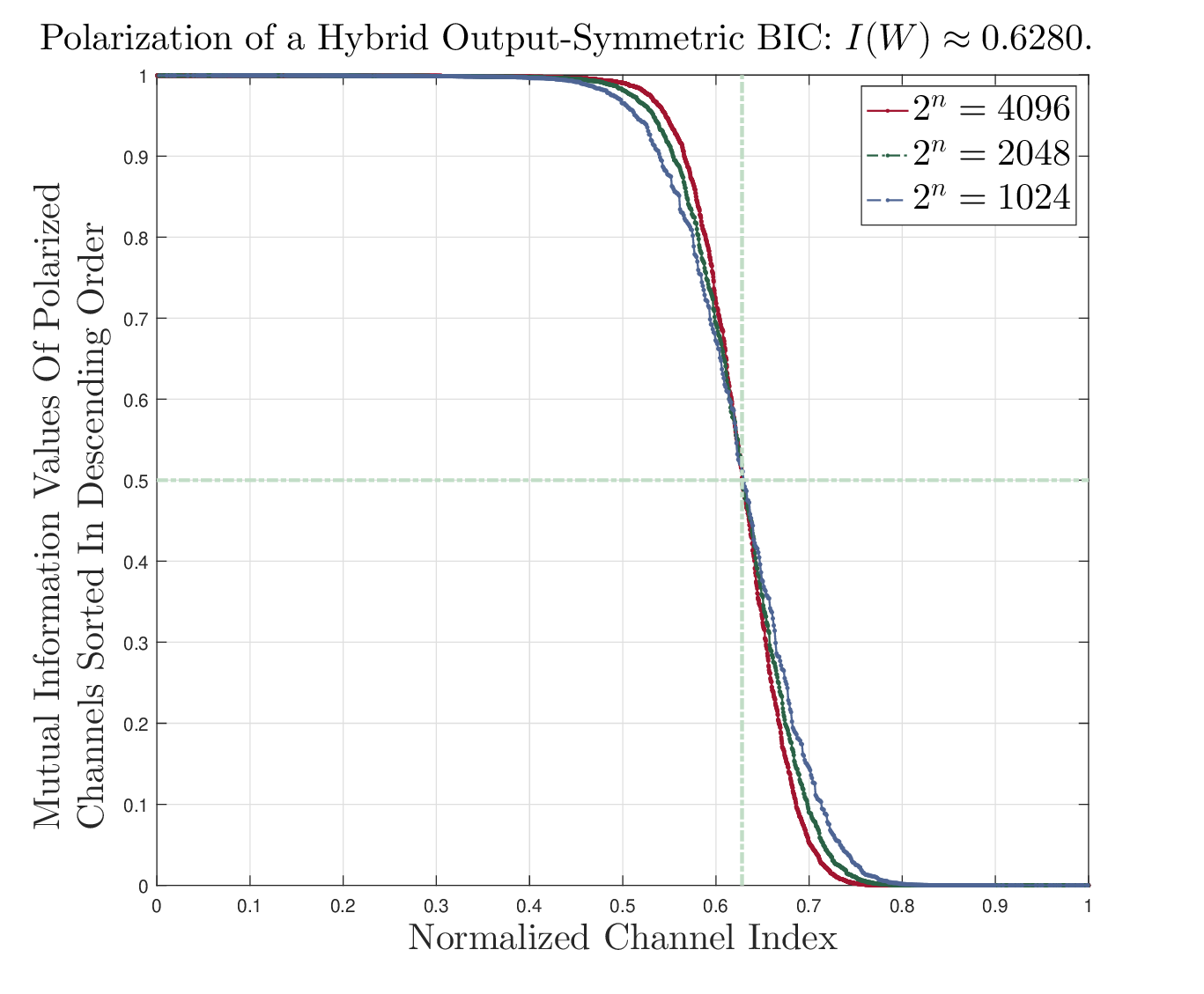}
}

\end{pspicture}
}

\end{center}
\caption{Experimental results for the polarization of a hybrid output-symmetric BIC with parameters $\epsilon_0 = 0.12$ and $\gamma_0 = 0.05$, with capacity $I(W) = (1-\epsilon_0)(1-h_2(\gamma_0)) \approx 0.6280$. Polar codes of block lengths $2^{n}$ were constructed for $n = 10, 11, 12$.} \label{fig:HybridBICExperimentalResults}
\end{figure}

Based on Theorem~\ref{thm:BM_polarization}, Corollary~\ref{corollary:BlackwellMeasuresPolarizedBSC} and Corollary~\ref{corollary:BIOSPolarization} are specialized for the class of symmetric $\BIC$s. Channel decompositions greatly simplify iterative computations and the construction of polar codes. A simple algorithm constructs all $2^{n}$ channels, as defined in Definition~\ref{def:SuccessivePolarization}, over $n$ iterations of the polar transform.

Consider two symmetric $\BIC$s $(\sY_1, W_1)$ and $(\sY_2, W_2)$. Due to Theorem~\ref{thm:BIOSdecomp}, there exist positive integers $m, k$, probability vectors $\lambda = (\lambda_1,\ldots,\lambda_m)$, $\mu = (\mu_1,\ldots,\mu_k)$, and parameters $p_1,\ldots,p_m \in [0,\frac{1}{2}]$, $q_1,\ldots,q_k \in [0,\frac{1}{2}]$ such that $W_1 \equiv \bigoplus^m_{i=1} \lambda_i \BSC(p_i)$ and $W_2 \equiv  \bigoplus^k_{j=1} \mu_j \BSC(q_j)$. In addition, we define the following parameters from Corollary~\ref{corollary:BlackwellMeasuresPolarizedBSC} and Corollary~\ref{corollary:BIOSPolarization}:
\begin{align}
\alpha_{ij} & := \frac{p_i q_j}{1 - p_i \star q_j}, \notag  \\
\beta_{ij} & := \frac{\bar{p}_i q_j}{p_i \star q_j}, \notag
\end{align}
where we define $(\alpha_{ij}, \beta_{ij}) := (0,0)$ if $(p_i, q_j) = (0, 0)$ in the degenerate case. Consider the sets $\mathcal{C}_{W_1}$ and $\mathcal{C}_{W_2}$ defined in Eqn.~\eqref{eq:Blackwell_set}:
\begin{IEEEeqnarray}{rClCll}
\cC_{W_1} & = & \{(\lambda_i, p_i) & : & i \in [m] & \}. \notag \\
\cC_{W_2} & = & \{(\mu_j, q_j)     & : & j \in [k] & \}. \notag
\end{IEEEeqnarray}
After one iteration of polarization, $W_1 \ATbad W_2$ and $W_1 \ATgood W_2$ have corresponding sets:
\begin{IEEEeqnarray}{rlCl}
& \IEEEeqnarraymulticol{3}{l}{ \cC_{W_1 \ATbad W_2}  = \Bigl\{ (\lambda_i \mu_j, p_i \star q_j) ~: i \in [m], j \in [k]\Bigl\}. } \label{eqn:BlackwellSetsPolarizedATbad} \\
& \cC_{W_1 \ATgood W_2} = & & \nonumber \\
& \quad \quad \Bigl\{ (\lambda_i\mu_j(1 - p_i \star q_j), \alpha_{ij})  & : & i \in [m], j \in [k] \Bigl\} ~ \cup \nonumber \\
& \quad \quad \Bigl\{ (\lambda_i \mu_j (p_i \star q_j), \beta_{ij} \wedge \bar{\beta}_{ij}) & : & i \in [m], j \in [k] \Bigl\}. \label{eqn:BlackwellSetsPolarizedATgood}
\end{IEEEeqnarray}
By iteratively applying Eqn.~\eqref{eqn:BlackwellSetsPolarizedATbad} and Eqn.~\eqref{eqn:BlackwellSetsPolarizedATgood}, the polar transforms may be applied successively to obtain the representation $\mathcal{C}_{W_b}$ for any transformed output-symmetric $\BIC$ $W_{b}$ defined in Definition~\ref{def:SuccessivePolarization}, where binary vector $b \in \{0,1\}^{n+1}$. Eqn.~\eqref{eqn:BlackwellSetsPolarizedATbad} and Eqn.~\eqref{eqn:BlackwellSetsPolarizedATgood} constitute an exact algorithm.

\begin{remark} \em{ Since the size of the output alphabet of transformed discrete channels increases exponentially with the number of iterations $n$, the above algorithm must be modified slightly to maintain computational tractability. One approach is to shift and merge the probability masses corresponding to the Dirac measures of transformed channels. The merge operation combines point masses located within the same interval of quantization. We refer to the literature which provides detailed analyses of channel approximation methods for constructing polar codes~\cite{pedarsani2011,tal_vardy_construct_2013}. Furthermore, as shown in~\cite{MHU_sublinear_construction_2019}, the construction of polar codes is even attainable with sub-linear complexity. }
\end{remark}

\subsection{Experimental results}

To corroborate Eqn.~\eqref{eqn:BlackwellSetsPolarizedATbad} and Eqn.~\eqref{eqn:BlackwellSetsPolarizedATgood}, experimental evidence is provided herein regarding the successive quantization and polarization of a hybrid output-symmetric $\BIC$. The hybrid $\BIC$ is a combination of $\BSC$ and $\BEC$ channels.

\begin{example}[Successive polarization of a hybrid output-symmetric $\BIC$]\label{example:HybridBIC} {\em Consider a $\BIC$ $(\sY, W)$ with output alphabet $\sY = \{0, 1, \texttt{e}\}$, and channel transition probabilities
\begin{align}
& W(\texttt{e}|0) = W(\texttt{e}|1) = \varepsilon_0, \notag \\
& W(0|0) = W(1|1) = (1-\varepsilon_0)(1-\GAM_0), \notag \\
& W(1|0) = W(0|1) = (1 - \varepsilon_0)\GAM_0. \notag
\end{align}
For $\varepsilon_0 = 0.12$, $\GAM_0 = 0.05$, the capacity $I(W) = (1-\varepsilon_0)(1 - h_2(\GAM_0)) \approx 0.6280$. Figure~\ref{fig:HybridBICExperimentalResults} depicts the mutual information values of transformed channels sorted in descending order after $n = 10, 11, 12$ levels of successive polarization. Channel quantization was applied to maintain computational tractability. Point masses corresponding to the Dirac measures of transformed channels were merged within dyadic intervals of length $2^{-L}$ with $L = 14$.

}
\end{example}

\section{One-step polarization of channel functionals $\I_f$}

Informally speaking, the polar transform \eqref{eq:polar_transform} replaces the original pair of $\BIC$s $W_1$ and $W_2$ with another pair, where $W_1 \ATbad W_2$ is ``worse'' than both $W_1$ and $W_2$, and $W_1 \ATgood W_2$ is ``better'' than both $W_1$ and $W_2$. The following definition makes precise the notion of one-step polarization of real-valued functionals for a class of channels.

\begin{definition}[One-step polarization of real-valued functionals]\label{def:PolarizationOfChannelFunctionals} Let $\cW$ denote a class of $\BIC$s. A channel functional $\Psi$ associates a real number $\Psi(W)$ to every $W \in \cW$. The functional $\Psi$ {\em polarizes} in one iteration on $\cW$ due to \Arikan's polar transform if, for any two $\BIC$s $W_1,W_2 \in \cW$,
\begin{align}
\Psi(W_1 \ATbad W_2) & \le \Psi(W_1) \wedge \Psi(W_2) \notag \\
& \le \Psi(W_1) \vee \Psi(W_2) \le \Psi(W_1 \ATgood W_2). \notag
\end{align}
This definition assumes that both $W_1 \ATbad W_2 \in \cW$ and $W_1 \ATgood W_2 \in \cW$.
\end{definition}

\subsection{One-step polarization of real-valued channel functionals}\label{sec:PolarizationBroadClassOfFunctionals}


In this section, we provide a direct proof from first principles, verifying that one-step polarization as defined in Definition~\ref{def:PolarizationOfChannelFunctionals} holds for many real-valued channel functionals on the class of symmetric $\BIC$s. As we detail in subsequent remarks, the following theorem has been rediscovered several times in the modern literature. Aside from its direct proof, it also follows as a consequence of the Blackwell--Sherman--Stein theorem.

\begin{theorem}[One-step polarization of real-valued functionals]\label{thm:BroadClassOfFunctionalsPolarize} All channel functionals $\I_f$ with a convex $f : [0,1] \to \Reals$ polarize in each iteration of the polar transform on the class of symmetric $\BIC$s. That is, if $W_1,W_2$ are two symmetric $\BIC$s, then
	\begin{align*}
	\I_f(W_1 \ATbad W_2) & \le \I_f(W_1) \wedge \I_f(W_2) \\
	                     & \le \I_f(W_1) \vee \I_f(W_2) \le \I_f(W_1 \ATgood W_2).
	\end{align*}
\end{theorem}

\begin{proof} Let $S_1 \sim \sm_{W_1}$ and $S_2 \sim \sm_{W_2}$ be independent. Then, using Theorem~\ref{thm:BM_polarization}, we can write
\begin{align}
	& \I_f(W_1 \ATbad W_2) \notag \\
    & \quad = \int_{[0,1]} f \d\sm_{W_1 \ATbad W_2} \notag \\
	& \quad = \int_{[0,1]} f \d(\sm_{W_1} \ATbad \sm_{W_2}) \notag \\
	& \quad = \E[f(S_1 S_2 + (1-S_1)(1-S_2))] \notag \\
	& \quad = \E\left[\E\left[f(S_1 S_2 + (1-S_1)(1-S_2)) \bigl| S_2 \right]\right] \notag \\
    & \quad \leq \E\left[ \E\left[S_2 f(S_1) + (1-S_2) f(1-S_1) \bigl| S_2\right]\right] \label{eqn:ThmAllFunctionalsPolarizeProofStepA}\\
	& \quad = \frac{1}{2}\E[f(S_1)] + \frac{1}{2}\E[f(1-S_1)] \label{eqn:ThmAllFunctionalsPolarizeProofStepB}\\
	& \quad = \I_f(W_1) \label{eqn:ThmAllFunctionalsPolarizeProofStepC},
\end{align}
where~\eqref{eqn:ThmAllFunctionalsPolarizeProofStepA} is by Jensen's inequality,~\eqref{eqn:ThmAllFunctionalsPolarizeProofStepB} follows from the fact that $S_1$ and $S_2$ are independent with $\E[S_1]=\E[S_2]=\frac{1}{2}$, and~\eqref{eqn:ThmAllFunctionalsPolarizeProofStepC} follows from the symmetry of $W_1$, which is equivalent to $\cL(S_1) = \cL(1-S_1)$. This shows that $\I_f(W_1 \ATbad W_2) \le \I_f(W_1)$. Conditioning on $S_1$ instead of $S_2$, we prove that $\I_f(W_1 \ATbad W_2) \le \I_f(W_2)$.
	
Using Theorem~\ref{thm:BM_polarization} and Jensen's inequality, we obtain
\begin{align*}
	\I_f(W_1 \ATgood W_2) & = \int_{[0,1]} f \d\sm_{W_1 \ATgood W_2} \\
	& = \int_{[0,1]} f \d(\sm_{W_1} \ATgood \sm_{W_2}) \\
	& = \E\Bigg[(1-S_1 \star S_2)f\left(\frac{S_1 S_2}{1-S_1 \star S_2}\right)\\
	& \qquad +(S_1 \star S_2)f\left(\frac{\bar{S}_1S_2}{S_1 \star S_2}\right)\Bigg] \\
	& \ge \E\left[f(S_2)\right] \notag \\
    & = \I_f(W_2). \notag
\end{align*}
By symmetry, $W_1 \ATgood W_2 \equiv W_2 \ATgood W_1$, so we also have $\I_f(W_1 \ATgood W_2) \ge \I_f(W_1)$.
\end{proof}
\begin{corollary}[Blackwell ordering of channels]\label{cor:BlackwellOrderingOfChannels} If $W_1, W_2$ are two symmetric $\BIC$s, then
\begin{align*}
	W_1 \ATbad W_2 \preceq_{{\rm B}} W_1 &\preceq_{{\rm B}} W_1 \ATgood W_2, \\
	W_1 \ATbad W_2 \preceq_{{\rm B}} W_2 &\preceq_{{\rm B}} W_1 \ATgood W_2.
\end{align*}
\end{corollary}
\begin{proof} Corollary~\ref{cor:BlackwellOrderingOfChannels} follows as a direct consequence of Theorem~\ref{thm:BroadClassOfFunctionalsPolarize} by invoking the Blackwell--Sherman--Stein theorem (Theorem~\ref{thm:BSS}).
\end{proof}

\begin{remark}[Alternative proof of Theorem~\ref{thm:BroadClassOfFunctionalsPolarize}] \em{ Theorem~\ref{thm:BroadClassOfFunctionalsPolarize} is established via a direct proof. Alternatively, as noted during the review of the present paper, one could observe that the Blackwell ordering of Corollary~\ref{cor:BlackwellOrderingOfChannels} holds for all symmetric $\BIC$s. Theorem~\ref{thm:BroadClassOfFunctionalsPolarize} then follows as a direct consequence of the Blackwell--Sherman--Stein theorem. This is true due to the presence of the ``if and only if'' statement (see Theorem~\ref{thm:BSS}). }
\end{remark}

\begin{remark}[Special cases of Theorem~\ref{thm:BroadClassOfFunctionalsPolarize}] \em{ From the induced functionals listed in Table~\ref{tbl:FunctionalsOfBICs}, the capacity $\I_f(W) = I(W)$ with $f(s) = 1 - h_2(s)$ was shown by \Arikan\ to exhibit the property of one-step polarization~\cite{arikan09}. An analogous result was shown for the Bhattacharyya parameter $\I_f(W) = -Z(W)$ with $f(s) = -2\sqrt{s(1-s)}$ in~\cite{arikan09}, and for Gallager's $E_0(\rho, W)$ parameter in~\cite[Lemma~4.5]{alsan_phd14}. Prior results emerge as special cases of Theorem~\ref{thm:BroadClassOfFunctionalsPolarize}. Moreover, Theorem~\ref{thm:BroadClassOfFunctionalsPolarize} implies that the property of one-step polarization holds for the Bayes error functionals $B_\lambda(W)$ and the squared maximal correlation $\rho^{2}_{\rm{max}}(W)$ described in Sec.~\ref{sec:InducedFunctionalsOfBICs}. To the best of our knowledge, such results have not been discussed previously. }
\end{remark}

\begin{remark}[Related proofs in modern literature] \em{ Related proofs linking the stochastic dominance and ordering of channels to the one-step polarization of functionals $\I_f$ for convex $f$ appear in~\cite[Ch.~4]{modern_coding_theory_book_2008} and~\cite[Lemma~6.16]{alsan_phd14}. As also described in Sec.~\ref{subsec:BlackwellMeasuresSymmetricBICs}, it is known that the Blackwell ordering is equivalent to the \textit{symmetric convex ordering} as introduced by Alsan in~\cite[Ch.~6]{alsan_phd14}. These independently-discovered results do not invoke the Blackwell--Sherman--Stein theorem. }
\end{remark}

\subsection{One-step polarization of $\I_f$: convex vs\ non-convex $f$}
All channel functionals $\I_f$ for convex $f$ in Table~\ref{tbl:FunctionalsOfBICs} polarize in each iteration on the class of output-symmetric $\BIC$s. This phenomenon is not guaranteed for non-convex $f$. Consider $f(s) = \psi_r(s)$, and $\I_f(W) = M_r(W)$, which represents the higher-order moments of the information density. As defined in Sec.~\ref{sec:InducedFunctionalsOfBICs}, $\psi_r(s) := s(1 + \log_2 s)^{r} + \bar{s} (1 + \log_2 \bar{s} )^{r}$. Figure~\ref{fig:PSI_FUNCTIONS} depicts that $\psi_1(s) = 1 - h_2(s)$ is convex on $[0, 1]$, and $\psi_2(s) = s(1 + \log_2 s)^{2} + \bar{s}(1 + \log_2 \bar{s} )^{2}$ is non-convex on $[0, 1]$. In the following example, it is shown that $M_2(W)$ does not polarize under \Arikan's transform, as according to Definition~\ref{def:PolarizationOfChannelFunctionals}.


\begin{figure}[t]
\begin{center}


\scalebox{1.0} 
{
\begin{pspicture}(0,-3.42)(5.2028127,3.82)

\rput(2.82, -0.05) {
\includegraphics[scale=0.512]{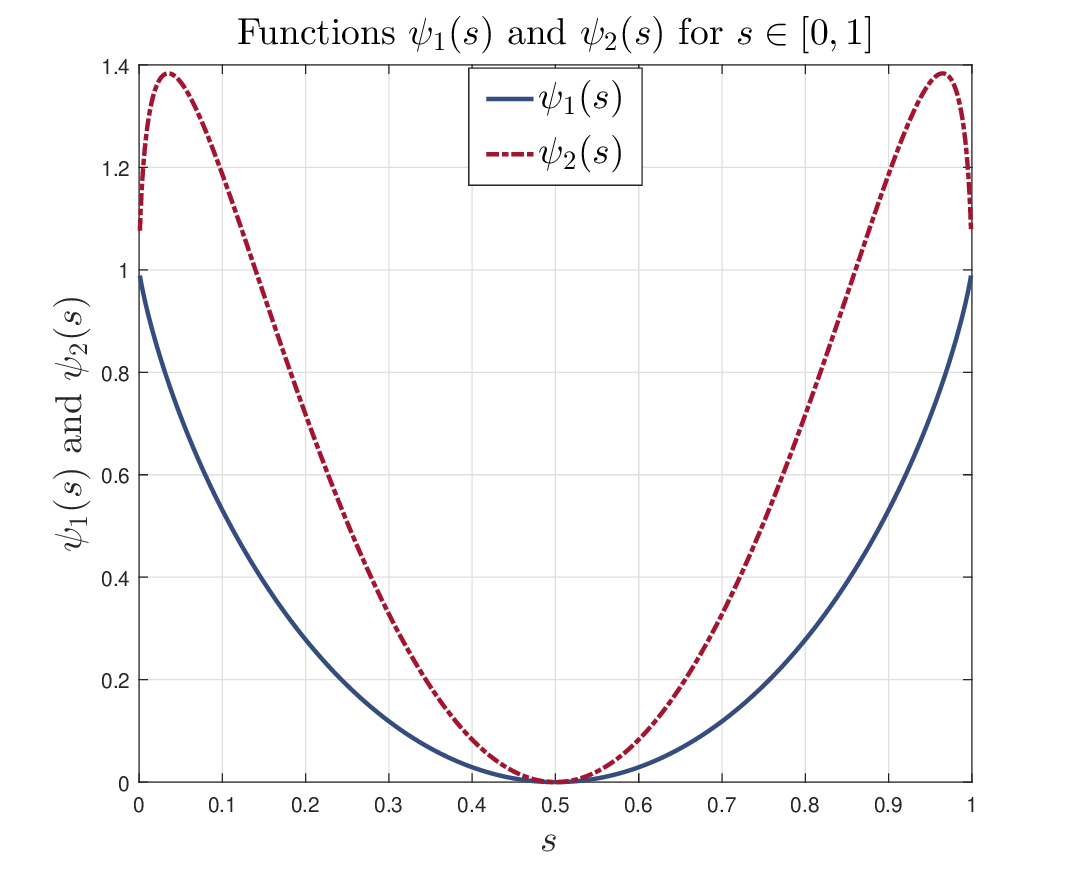}
}

\end{pspicture}
}

\end{center}
\caption{Functions $\psi_1(s)$ and $\psi_2(s)$ for $s \in [0,1]$.} \label{fig:PSI_FUNCTIONS}
\end{figure}

\begin{example}[Counter-example for one-step polarization for non-convex $f$] {\em Let $W = \BSC(\GAM)$ and consider the channels $W \ATbad W$ and $W \ATgood W$. Applying Lemma~\ref{lemma:StructuralPolarizedBSC},
\begin{alignat}{3}
& \BSC(\GAM) \ATbad \BSC(\GAM) && \equiv \BSC(\GAM \star \GAM), && \notag \\
& \BSC(\GAM) \ATgood \BSC(\GAM) && \equiv (1-\GAM \star \GAM) && \BSC\left(\frac{\GAM^{2}}{1 - \GAM \star \GAM}\right) \notag \\
& \quad \quad \quad && \quad \quad \oplus (\GAM \star \GAM) && \BSC\left(\frac{1}{2}\right). \notag
\end{alignat}
Due to Corollary~\ref{cor:BlackwellOrderingOfChannels}, the Blackwell ordering holds: $W \ATbad W \preceq_{{\rm B}} W \preceq_{{\rm B}} W \ATgood W$. However, the second moments $M_2(W)$, $M_2(W \ATbad W)$, and $M_2(W \ATgood W)$ do not satisfy Definition~\ref{def:PolarizationOfChannelFunctionals} for polarization. The following values are computed via Eqn.~\eqref{eqn:InducedFunctionalSymmBIC1}, and also Eqn.~\eqref{eqn:M2MomentForBSC}. For $p = 0.05$,
\begin{align}
M_2(\BSC(0.05)) = 1.3664, \notag \\
M_2(\BSC(0.05) \ATbad \BSC(0.05)) = 1.2085, \notag \\
M_2(\BSC(0.05) \ATgood \BSC(0.05)) = 1.0359. \notag
\end{align}
It is evident that $M_2(W \ATgood W)$ is not greater than $M_2(W)$. The functional $\I_f(W) = M_2(W)$ does not exhibit one-step polarization, due to the non-convexity of $f(s) = \psi_2(s)$.}
\end{example}

\section{Properties of the polarization process}\label{sec:PolarizationRandomProcesses}

An important method of analyzing the successive polarization of channels is through a certain random process referred to as the \emph{polarization process}. Starting from a $\BIC$ $(\sY, W)$, one level of polarization yields either $W \ATbad W$ or $W \ATgood W$. As noted by \Arikan~\cite{arikan09}, a random path over $n$ levels of polarization leads to randomly selecting a channel $W_b$ where $b \in \{0,1\}^{n+1}$ in Definition~\ref{def:SuccessivePolarization}:

\begin{definition}[Channel polarization --- random processes]\label{def:RandomProcessOfPolarization}
Consider a $\BIC$ $(\sY, W)$. Let $\{B_n\}^\infty_{n=1}$ be a sequence of i.i.d.\ $\Bernoulli(1/2)$ random variables. Let $W_0 = W$, and
$$
W_n = \begin{cases}
W_{n-1} \ATbad W_{n-1}, & \text{if $B_n = 0$} \\
W_{n-1} \ATgood W_{n-1}, & \text{if $B_n = 1$}
\end{cases}
$$
for $n > 0$. Define the random processes $\{I_n\}^\infty_{n=0}$ and $\{Z_n\}^\infty_{n=0}$ via $I_n = I(W_n)$ and $Z_n = Z(W_n)$. In general, a random process $\{\I_f(W_n)\}_{n=0}^{\infty}$ is obtained for any induced functional listed in Table~\ref{tbl:FunctionalsOfBICs}.
\end{definition}

\begin{example}[Properties of $\{I_n\}_{n=0}^{\infty}$ and $\{Z_n\}_{n=0}^{\infty}$] {\em
As shown in~\cite{arikan09}, for the class of output-symmetric $\BIC$s, $\{I_n\}$ is a nonnegative martingale, while $\{Z_n\}$ is a nonnegative supermartingale, both with respect to the natural filtration generated by $\{B_n\}$. More precisely,
\begin{align}
\E\left[ I_{n+1} \bigl| B_1, B_2, \ldots, B_n \right] & = I_{n}, \notag \\
\E\left[ Z_{n+1} \bigl| B_1, B_2, \ldots, B_n \right] & \le Z_{n}. \notag
\end{align}
In order to prove the above properties, consider any two $\BIC$s $(\sY, W)$ and $(\sY^{\prime}, W^{\prime})$ from the class of output-symmetric $\BIC$s. As first noted by~\cite{arikan09},
\begin{align}
I(W \ATbad W^{\prime}) + I(W \ATgood W^{\prime}) & = I(W) + I(W^{\prime}), \notag \\
Z(W \ATbad W^{\prime}) + Z(W \ATgood W^{\prime}) & \le Z(W) + Z(W^{\prime}). \notag
\end{align}
The first relation is due to the conservation of mutual information. The second relation is due to the fact that $Z(W \ATbad W^{\prime}) \leq Z(W) + Z(W^{\prime}) - Z(W)Z(W^{\prime})$ and $Z(W \ATgood W^{\prime}) = Z(W)Z(W^{\prime})$, a result first observed by~\Arikan~\cite[Prop.~5]{arikan09}. }
\end{example}

\subsection{The random processes $\{\I_f(W_n)\}_{n=0}^{\infty}$}

As detailed in Definition~\ref{def:RandomProcessOfPolarization}, a random polarization process $\{\I_f(W_n)\}_{n=0}^{\infty}$ is defined for any induced functional $\I_f$. In order to analyze the properties of the random process, the following relations are introduced:

\begin{definition}[$f$-relations]\label{def:PreservingImprovingDecreasingRelations} Consider two arbitrary $\BIC$s $(\sY, W)$ and $(\sY^{\prime}, W^{\prime})$ from a given class $\cW$ of $\BIC$s. Let $\cW$ be closed under the polar transform operations; i.e., both $W \ATbad W^{\prime} \in \cW$ and $W \ATgood W^{\prime} \in \cW$. Let $f : [0,1] \to \Reals$ be a continuous function. We say that the polarization process on the class $\cW$ is:
\begin{align}
& \mbox{$f$-preserving if for all $W, W^{\prime} \in \cW$,} \notag \\
& \quad \I_f(W \ATbad W^{\prime}) + \I_f(W \ATgood W^{\prime}) = \I_f(W) + \I_f(W^{\prime}); \label{eqn:fpreserving} \\
& \mbox{$f$-improving if for all $W, W^{\prime} \in \cW$,} \notag && \\
& \quad \I_f(W \ATbad W^{\prime}) + \I_f(W \ATgood W^{\prime}) \geq \I_f(W) + \I_f(W^{\prime}); \label{eqn:fimproving} \\
& \mbox{$f$-decreasing if for all $W, W^{\prime} \in \cW$,} \notag && \\
& \quad \I_f(W \ATbad W^{\prime}) + \I_f(W \ATgood W^{\prime}) \leq \I_f(W) + \I_f(W^{\prime}). \label{eqn:fdecreasing}
\end{align}
\end{definition}

Consider the above conditions and bounded random processes. If~\eqref{eqn:fpreserving} holds for all $W, W^{\prime} \in \cW$, then the random process $\{\I_f(W_n)\}_{n=0}^{\infty}$ is a martingale. If~\eqref{eqn:fimproving} holds, $\{\I_f(W_n)\}_{n=0}^{\infty}$ is a submartingale. Similarly, if~\eqref{eqn:fdecreasing} holds, $\{\I_f(W_n)\}_{n=0}^{\infty}$ is a supermartingale. If $\cW$ is the class of all output-symmetric $\BIC$s, the following theorem shows that it suffices to verify the $f$-relations only on the subclass consisting of $\BSC$s:

\begin{figure*}[!t]
  \centering
    \subfloat[$\cR_{{\rm NP}}\big(\frac{1}{2}W \oplus \frac{1}{2} W^{\prime}\big)$]{{\includegraphics[width=0.352\textwidth]{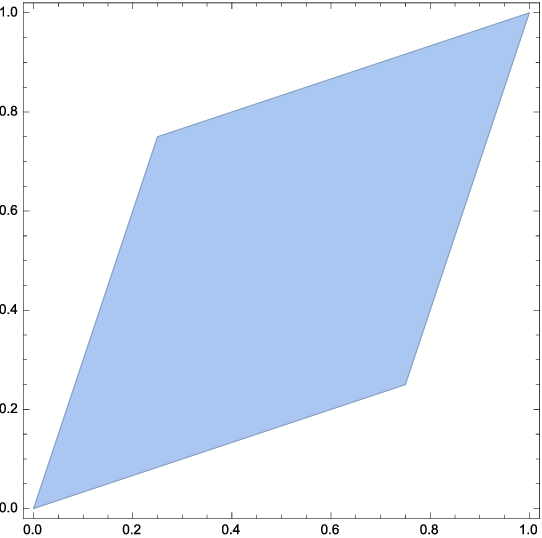} }}%
    \qquad \qquad
    \subfloat[${\cR}_{{\rm NP}}\big(\frac{1}{2}(W \ATbad W^{\prime}) \oplus \frac{1}{2}(W \ATgood W^{\prime})\big)$]{{\includegraphics[width=0.352\textwidth]{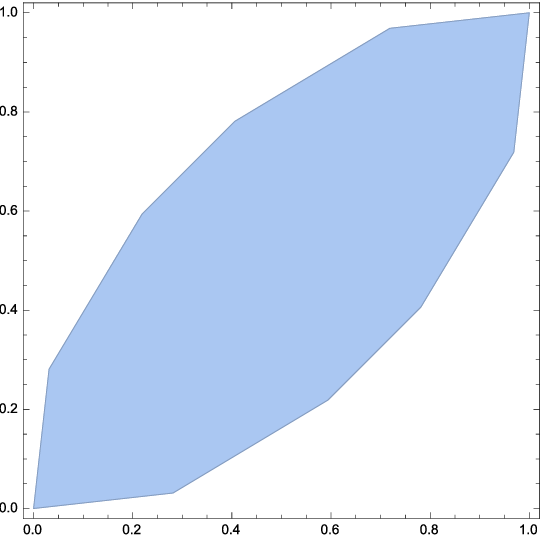} }}%
    \caption{The Neyman-Pearson regions considered in Example~\ref{example:CounterArgumentConjecture} in the case $W \equiv W^{\prime} \equiv \BSC(1/4)$.}%
    \label{fig:counterex}%
\end{figure*}

\begin{theorem}[$f$-relations for all symmetric $\BIC$s]\label{thm:fRelationsForOutputSymmetricBICs} The polarization process is $f$-preserving, $f$-improving, or $f$-decreasing as defined in Definition~\ref{def:PreservingImprovingDecreasingRelations} on the class of output-symmetric $\BIC$s if and only if \eqref{eqn:fpreserving},~\eqref{eqn:fimproving}, or~\eqref{eqn:fdecreasing} holds respectively for all pairs $(W, W^{\prime}) = (\BSC(\GAM),\BSC(\NU))$, $(\GAM, \NU) \in [0,\frac{1}{2}]\times [0,\frac{1}{2}]$.
\end{theorem}

\begin{proof}
Consider the $f$-improving relation for a class of $\BIC$s and induced functional $\I_f(\cdot)$. If~\eqref{eqn:fimproving} holds for all symmetric $\BIC$s, then it holds for all $\BSC$s. To prove the converse, fix two symmetric $\BIC$s $W, W^{\prime}$. By Theorem~\ref{thm:BIOSdecomp}, the following channel decompositions exist:
\begin{align}
W          & \equiv  \bigoplus^m_{i=1} \lambda_i \BSC(\GAM_i), \notag \\
W^{\prime} & \equiv  \bigoplus^k_{j=1} \mu_j \BSC(\NU_j). \notag
\end{align}
By Corollary~\ref{corollary:BIOSPolarization}, the Blackwell measures of the transformed channels $W \ATbad W^{\prime}$ and $W \ATgood W^{\prime}$ are given by
\begin{align}
\sm_{W \ATbad W^{\prime}} & = \sum^m_{i=1}\sum^k_{j=1}\lambda_i\mu_j \sm_{\BSC(\GAM_i \star \NU_j)}, \notag \\
\sm_{W \ATgood W^{\prime}}     & = \sum^m_{i=1}\sum^k_{j=1}\lambda_i\mu_j \sm_{\BSC(\GAM_i) \times \BSC(\NU_j)}. \notag
\end{align}
Consequently, using the definitions for induced functionals in Sec.~\ref{sec:InducedFunctionalsOfBICs}, and the assumption that~\eqref{eqn:fimproving} holds for all $\BSC$s, we have
\begin{align*}
	&\I_f(W \ATbad W^{\prime}) + \I_f(W \ATgood W^{\prime}) \\
	&=\int f \d\sm_{W \ATbad W^{\prime}} + \int f\d\sm_{W \ATgood W^{\prime}} \\
	&= \sum^m_{i=1}\sum^k_{j=1} \lambda_i \mu_j \int f\left(\d\sm_{\BSC(\GAM_i \star \NU_j)} +\d\sm_{\BSC(\GAM_i) \times \BSC(\NU_j)} \right) \\
	&= \sum^m_{i=1}\sum^k_{j=1} \lambda_i \mu_j \Big( \I_f(\BSC(\GAM_i \! \star \! \NU_j)) \!  +  \! \I_f(\BSC(\GAM_i) \!  \times \! \BSC(\NU_j)) \! \Big) \\
	&\ge \sum^m_{i=1}\sum^k_{j=1}\lambda_i \mu_j \Big(\I_f(\BSC(\GAM_i)) + \I_f(\BSC(\NU_j)) \Big) \\
	&= \sum^m_{i=1} \lambda_i \I_f(\BSC(\GAM_i)) + \sum^k_{j=1}\mu_j \I_f(\BSC(\NU_j)) \\
	&= \I_f(W) + \I_f(W^{\prime}). 
\end{align*}
Theorem~\ref{thm:fRelationsForOutputSymmetricBICs} is established in an identical manner for the $f$-preserving and $f$-decreasing relations.
\end{proof}

The following corollary follows directly from Theorem~\ref{thm:fRelationsForOutputSymmetricBICs} and Lemma~\ref{lemma:StructuralPolarizedBSC}. Consider if $f(s) = f(\bar{s})$. Then $\I_f(\BSC(p)) = f(p)$, and Theorem~\ref{thm:fRelationsForOutputSymmetricBICs} may be presented in a simplified form via functional inequalities. We note that Theorem~\ref{thm:fRelationsForOutputSymmetricBICs} can still be written via functional inequalities without assuming $f(s) = f(\bar{s})$.

\begin{corollary}[Functional inequalities]\label{corollary:NewMartingales}
Let $f: [0,1] \rightarrow \mathbb{R}$ be a continuous function. Define the following gap function for a particular $f$, and $p, q \in [0, \frac{1}{2}]$:
\begin{IEEEeqnarray}{rCl}
\textsc{gap}_{f}(p, q) & := & \I_f(\BSC(p)) + \I_f(\BSC(q)) \notag \\
& & -\> \I_f(\BSC(p) \ATbad \BSC(q)) \notag \\
& & -\> \I_f(\BSC(p) \ATgood \BSC(q)). \label{eqn:OriginalDefinitionGapFunction}
\end{IEEEeqnarray}
\noindent As a corollary to Theorem~\ref{thm:fRelationsForOutputSymmetricBICs}, the following criterion correspond to Eqns.~\eqref{eqn:fpreserving},~\eqref{eqn:fimproving}, and~\eqref{eqn:fdecreasing}. Over the class of symmetric $\BIC$s, the polarization process is:
\begin{IEEEeqnarray}{ll}
& \mbox{$f$-preserving if and only if for all $p, q  \in [0, \frac{1}{2}]$,} \notag \\
& \qquad \qquad \textsc{gap}_{f}(p, q) = 0; \label{eqn:fpreservingFunctionalIneq} \\
& \mbox{$f$-improving if and only if for all $p, q  \in [0, \frac{1}{2}]$,} \notag \\
& \qquad \qquad \textsc{gap}_{f}(p, q) \leq 0; \label{eqn:fimprovingFunctionalIneq} \\
& \mbox{$f$-decreasing if and only if for all $p, q \in [0, \frac{1}{2}]$,} \notag \\
& \qquad \qquad \textsc{gap}_{f}(p, q) \geq 0. \label{eqn:fdecreasingFunctionalIneq}
\end{IEEEeqnarray}
If $f(s) = f(\bar{s})$, then Eqn.~\eqref{eqn:OriginalDefinitionGapFunction} simplifies to:
\begin{IEEEeqnarray}{rCl}
\textsc{gap}_{f}(p, q) & = & f(p) + f(q) - f(p \star q) \notag \\
& & -\> (1 - p \star q)f(\alpha) - (p \star q)f(\beta). \label{eqn:DefinitionGapFunctionCorollary}
\end{IEEEeqnarray}
The parameters $\alpha := \frac{ \GAM \NU }{ 1 - \GAM \star \NU }$ and $\beta := \frac{ \bar{\GAM} \NU }{ \GAM \star \NU }$ were defined in Lemma~\ref{lemma:StructuralPolarizedBSC}, where $(p, q) \neq (0,0)$ in the degenerate case. Recall that $p \star q \in [0, \frac{1}{2}]$, $\alpha \in [0, \frac{1}{2}]$ and $\beta \wedge \bar{\beta} \in [0, \frac{1}{2}]$. Note that if $f(s) = f(\bar{s})$, then $f(\beta) = f(\bar{\beta}) = f(\beta \wedge \bar{\beta})$.
\end{corollary}

\subsection{$f$-relations: The case of convex $f$}

It is tempting to conjecture that an $f$-relation such as the $f$-improving relation given in Eqn.~\eqref{eqn:fimproving} holds for all convex $f$ on the class of output-symmetric $\BIC$s. However, the following counter-example proves that this conjecture is false. Similarly, it is tempting to think that the convexity of $f$ would greatly simplify the functional inequalities in Corollary~\ref{corollary:NewMartingales}. However, this is not the case. Whether $f$ is convex (or non-convex) does not directly imply an $f$-relation.

\begin{example}[Counter-argument for conjecture]\label{example:CounterArgumentConjecture} {\em
Suppose that the $f$-improving relation in Eqn.~\eqref{eqn:fimproving} were true for all convex $f$. More precisely, assume that $\frac{1}{2}\I_f(W \ATbad W^{\prime}) + \frac{1}{2}\I_f(W \ATgood W^{\prime}) \geq \frac{1}{2}\I_f(W) + \frac{1}{2}\I_f(W^{\prime})$ for all convex $f$. According to Theorem~\ref{thm:BSS}, the Blackwell--Sherman--Stein theorem, the channel $\frac{1}{2}(W \ATbad W^{\prime}) \oplus \frac{1}{2}(W \ATgood W^{\prime})$ would dominate the channel $\frac{1}{2}W \oplus \frac{1}{2}W^{\prime}$. However, this conjecture turns out to be false. Figure~\ref{fig:counterex} shows the Neyman--Pearson regions of $\frac{1}{2}(W \ATbad W^{\prime}) \oplus \frac{1}{2}(W \ATgood W^{\prime})$ and $\frac{1}{2}W \oplus \frac{1}{2}W^{\prime}$ when $W \equiv W^{\prime} \equiv \BSC(1/4)$. It is evident that the latter is not a subset of the former. By Theorem~\ref{thm:NeymanPearsonCriterion}, the Neyman--Pearson criterion for Blackwell dominance, we arrive at a contradiction.}
\end{example}

\begin{example}[Counter-example for convex $f$] {\em Consider $f_\lambda(s) = \bar{\lambda} \wedge \lambda - (2\bar{\lambda} s) \wedge (2\lambda\bar{s})$ for $\lambda \in [0,1]$. Specifically, consider $\lambda = \frac{1}{3}$, and $f \equiv f_{1/3}(s)$. In this case, $f(s)$ is convex, and $f(s) \neq f(\bar{s})$. Define the following real-valued function for two $\BIC$s $W$ and $W'$:
\begin{IEEEeqnarray}{rCl}
\zeta(W, W') & := & \I_{f}(W) + \I_{f}(W') \notag \\
& & -\> \I_{f}(W \ATbad W') - \I_{f}(W \ATgood W'). \notag
\end{IEEEeqnarray}
A straightforward calculation yields:
\begin{IEEEeqnarray*}{rCl}
\zeta(\BSC(1/4), \BSC(1/4)) & > & 0, \\
\zeta(\BSC(3/8), \BSC(3/8)) & < & 0.
\end{IEEEeqnarray*}
Equivalently, as summarized by Eqn.~\eqref{eqn:OriginalDefinitionGapFunction}, $\textsc{gap}_{f}(p, q) > 0$ for $p = q = \frac{1}{4}$, but $\textsc{gap}_{f}(p, q) < 0$ for $p = q = \frac{3}{8}$. Thus, although $f_{1/3}(s)$ is a convex function, the $f$-relations of Definition~\ref{def:PreservingImprovingDecreasingRelations} do not hold consistently for all $\BSC$ pairs. According to Theorem~\ref{thm:fRelationsForOutputSymmetricBICs}, the random process $\{\I_f(W_n)\}_{n=0}^{\infty}$ is not a martingale, submartingale, or supermartingale for the class of symmetric $\BIC$s.}
\end{example}


\subsection{$f$-relations: The case of non-convex $f$}
Consider $f \equiv \psi_2(s)$ where $\psi_2(s) := s(1 + \log_2(s))^{2} + \bar{s}(1 + \log_2 \bar{s})^{2}$. As plotted in Figure~\ref{fig:PSI_FUNCTIONS}, $f$ is a non-convex function, and $f(s) = f(\bar{s})$. As proven in Sec.~\ref{subsec:RthMomentInformationDensity}, $\I_f(W) = M_2(W)$, which is the second moment of information density. The following example shows numerically that the corresponding polarization process satisfies the $f$-decreasing relation of Eqn.~\eqref{eqn:fdecreasing} for all symmetric $\BIC$s.
\begin{example}[Polarization process for $M_2(W)$]\label{example:M2Gap} {\em
Consider $f \equiv \psi_2(s)$. In order to prove the $f$-decreasing property of Eqn.~\eqref{eqn:fdecreasing}, the following must hold for all pairs of symmetric $\BIC$s $(\sY, W)$ and $(\sY^{\prime}, W^{\prime})$:
\begin{align}
M_2(W \ATbad W^{\prime}) + M_2(W \ATgood W^{\prime}) & \leq M_2(W) + M_2(W^{\prime}). \notag
\end{align}
According to Theorem~\ref{thm:fRelationsForOutputSymmetricBICs}, it suffices to consider the space of $\BSC$s. Therefore, consider $W \equiv \BSC(p)$ and $W^{\prime} \equiv \BSC(q)$ for all $p, q \in [0, \frac{1}{2}]$ where $(p, q) \neq (0, 0)$. Then Eqn.~\eqref{eqn:OriginalDefinitionGapFunction} takes the following form for $f \equiv \psi_2(s)$:
\begin{IEEEeqnarray}{rCl}
\textsc{gap}_{f}(p, q) & = & M_2(\BSC(p)) + M_2(\BSC(q)) \notag \\
& & -\> M_2(\BSC(p) \ATbad \BSC(q)) \notag \\
& & -\> M_2(\BSC(p) \ATgood \BSC(q)). \notag
\end{IEEEeqnarray}
For $f \equiv \psi_2(s)$, $f(s) = f(\bar{s})$, and Corollary~\ref{corollary:NewMartingales} is applicable. It remains to prove the functional inequality of Eqn.~\eqref{eqn:fdecreasingFunctionalIneq} given in Corollary~\ref{corollary:NewMartingales}. The gap function simplifies to Eqn.~\eqref{eqn:DefinitionGapFunctionCorollary}, which is written as follows for $f \equiv \psi_2(s)$:
\begin{IEEEeqnarray}{rCl}
\textsc{gap}_{f}(p, q) & = & \psi_2(p) + \psi_2(q) - \psi_2(p \star q) \notag \\
& & -\> (1 - p \star q)\psi_2(\alpha) - (p \star q)\psi_2(\beta). \label{eqn:Psi2GapFunctionSpecialized}
\end{IEEEeqnarray}
The parameters $\alpha := \frac{ \GAM \NU }{ 1 - \GAM \star \NU }$ and $\beta := \frac{ \bar{\GAM} \NU }{ \GAM \star \NU }$ were defined in Lemma~\ref{lemma:StructuralPolarizedBSC}. Note that $\textsc{gap}_{f}(\GAM, \NU) = \textsc{gap}_{f}(\NU, \GAM)$. In addition, along the boundaries, $\textsc{gap}_{f}(\GAM, 0) = 0$, $\textsc{gap}_{f}(0, \NU) = 0$, $\textsc{gap}_{f}(\GAM, \frac{1}{2}) = 0$, and $\textsc{gap}_{f}(\frac{1}{2}, \NU) = 0$. Figure~\ref{fig:GapResults} provides numerical evidence that $\textsc{gap}_{f}(\GAM, \NU) \geq 0$ for all $\BSC$ pairs. Theorem~\ref{thm:fRelationsForOutputSymmetricBICs} and Corollary~\ref{corollary:NewMartingales} imply the same result holds for all symmetric $\BIC$s.}
\end{example}

\begin{remark} \em{ Theorem~\ref{thm:fRelationsForOutputSymmetricBICs} and Corollary~\ref{corollary:NewMartingales} provide a feasible numerical approach to analyze all channel functionals $\I_f$ listed in Table~\ref{tbl:FunctionalsOfBICs}, as well as additional functionals not listed. While Example~\ref{example:M2Gap} shows numerically that $\textsc{gap}_{f}(\GAM, \NU) \geq 0$, in this particular example, the exact analytical proof of the inequality requires a lengthy analysis. We refer the reader to the analytical proofs by \Arikan\ dedicated to the following result that is closely related~\cite[Theorem~1]{arikan_varentropy_2016}:
\begin{align}
V(W \ATbad W^{\prime}) + V(W \ATgood W^{\prime}) & \leq V(W) + V(W^{\prime}). \label{eqn:VarEntropyDecreases}
\end{align}
The channel dispersion $V(W) := M_2(W) - (I(W))^{2}$ assumes a uniform input distribution. The result by \Arikan\ in~\cite{arikan_varentropy_2016} for {\em varentropy} applies more generally to so-called {\em binary data elements} incorporating the input distribution as well as each channel's conditional distribution.
}
\end{remark}

\section{A new supermartingale}\label{sec:NewSubSuperMartingales}

The necessary and sufficient condition given in Theorem~\ref{thm:fRelationsForOutputSymmetricBICs} provides a helpful reduction from the space of symmetric $\BIC$s to the space of $\BSC$s, in order to prove the existence of new submartingales and supermartingales. Deriving one of the functional inequalities of Corollary~\ref{corollary:NewMartingales} is a method of proof. Using this method, we prove that the polarization process associated with the squared Hirschfeld-Gebelein-R\'enyi maximal correlation parameter $\rho^{2}_{\rm{max}}(W)$ is a supermartingale.

\begin{figure}[t]
\begin{center}


\scalebox{1.0} 
{
\begin{pspicture}(0,-2.82)(5.92028127,2.82)

\rput(2.82, -0.25) {
\includegraphics[scale=0.52]{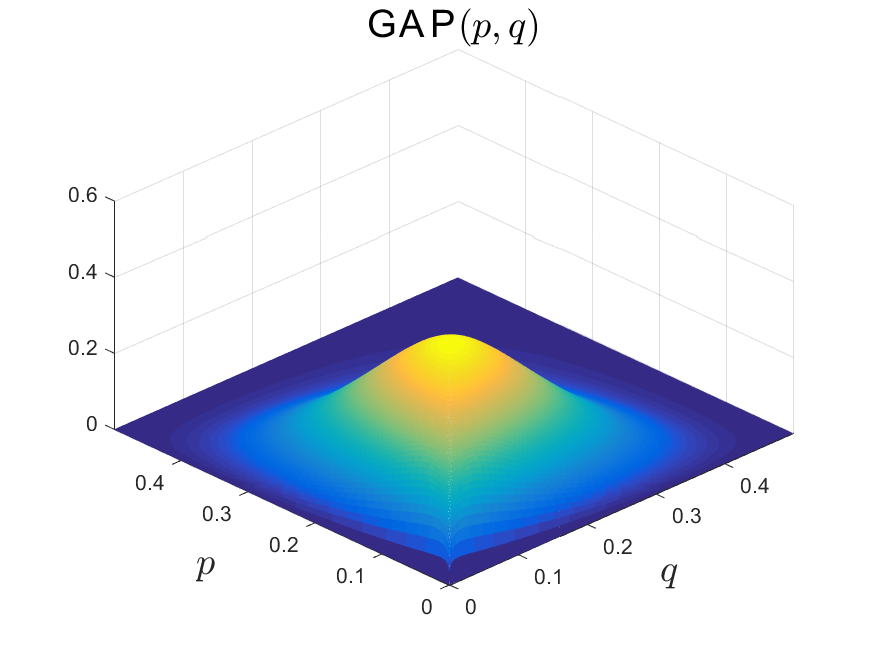}
}

\end{pspicture}
}

\end{center}
\caption{ The symmetric gap function $\textsc{gap}_{f}(\GAM, \NU)$ with $f \equiv \psi_2(s)$, as detailed in Eqn.~\eqref{eqn:Psi2GapFunctionSpecialized} of Example~\ref{example:M2Gap}. } \label{fig:GapResults}
\end{figure}


\begin{theorem}[Supermartingale associated with $\rho^{2}_{\rm{max}}(W)$]\label{thm:MaximalCorrelationSuperMG}
Consider $f \equiv f(s) = (2s - 1)^{2}$. As proven in Sec.~\ref{subsec:SquaredMaxiximalCorrelation}, for a $\BIC$ $(\sY, W)$, the induced functional $\I_f(W) = \rho^{2}_{\rm{max}}(W)$. For the class of symmetric $\BIC$s, the polarization process $\{\I_f(W_n)\}_{n=0}^{\infty}$ defined in Definition~\ref{def:RandomProcessOfPolarization} is a supermartingale, and converges almost surely on the interval $[0, 1]$.
\end{theorem}

\begin{proof} In order to prove that $\{\I_f(W_n)\}_{n=0}^{\infty}$ is a supermartingale, the $f$-decreasing property must hold for all pairs of symmetric $\BIC$s $(\sY, W)$ and $(\sY^{\prime}, W^{\prime})$:
\begin{align}
\rho^{2}_{\rm{max}}(W \ATbad W^{\prime}) + \rho^{2}_{\rm{max}}(W \ATgood W^{\prime}) & \leq \rho^{2}_{\rm{max}}(W) + \rho^{2}_{\rm{max}}(W^{\prime}). \notag
\end{align}
According to Theorem~\ref{thm:fRelationsForOutputSymmetricBICs}, it suffices to consider the space of $\BSC$s. Along these lines, consider $W \equiv \BSC(p)$ and $W^{\prime} \equiv \BSC(q)$ for all $p, q \in [0, \frac{1}{2}]$ where $(p, q) \neq (0, 0)$. Then Eqn.~\eqref{eqn:OriginalDefinitionGapFunction} takes the following form for $f \equiv f(s) = (2s-1)^{2}$:
\begin{IEEEeqnarray}{rCl}
\textsc{gap}_{f}(p, q) & = & \rho^{2}_{\rm{max}}(\BSC(p)) + \rho^{2}_{\rm{max}}(\BSC(q)) \notag \\
& & -\> \rho^{2}_{\rm{max}}(\BSC(p) \ATbad \BSC(q)) \notag \\
& & -\> \rho^{2}_{\rm{max}}(\BSC(p) \ATgood \BSC(q)). \notag
\end{IEEEeqnarray}
It remains to prove the functional inequality of Eqn.~\eqref{eqn:fdecreasingFunctionalIneq} given in Corollary~\ref{corollary:NewMartingales}. It is evident that $f(s) = f(\bar{s})$. Hence, the gap function simplifies to Eqn.~\eqref{eqn:DefinitionGapFunctionCorollary}, which is written as follows for $f \equiv f(s) = (2s-1)^{2}$:
\begin{IEEEeqnarray}{rCl}
\textsc{gap}_{f}(p, q) & = & f(p) + f(q) - f(p \star q) \notag \\
& & -\> (1 - p \star q)f(\alpha) - (p \star q)f(\beta). \label{eqn:GapFunctionMaxCorrDefinition}
\end{IEEEeqnarray}
The parameters $\alpha := \frac{ \GAM \NU }{ 1 - \GAM \star \NU }$ and $\beta := \frac{ \bar{\GAM} \NU }{ \GAM \star \NU }$ were defined in Lemma~\ref{lemma:StructuralPolarizedBSC}. After a series of algebraic simplifications provided in Appendix~\ref{appendix:ProofNewSuperMGMaximalCorrelations}, it can be shown that
\begin{align}
\textsc{gap}_{f}(p, q) & = \frac{4pq\bar{p}\bar{q}(2p-1)^{2} (2q-1)^{2}}{(1 - p \star q)(p \star q)}. \label{eqn:GapFunctionMaxCorrelation}
\end{align}
Both the numerator and denominator of Eqn.~\eqref{eqn:GapFunctionMaxCorrelation} are positive quantities. Hence, $\textsc{gap}(p, q) \geq 0$ for all $p, q \in [0, \frac{1}{2}]$. Since the squared maximal correlation parameter is bounded in the interval $[0, 1]$, by standard arguments for the convergence of supermartingales, the polarization process $\{\I_f(W_n)\}_{n=0}^{\infty}$ converges almost surely on $[0, 1]$.
\end{proof}

\begin{remark}[Alternative methods to prove convergence] \em{ As discussed in Sec.~\ref{subsec:RelatedPriorWork}, the martingale conditions are not required to prove the convergence of random processes. The authors of~\cite{alsan_telatar_simple_proof_2016} provide a simple proof of polarization that avoids the explicit use of martingales. }
\end{remark}

\begin{remark}[Maximal correlations] \em{ By applying alternative methods, the first author of the present paper showed in~\cite{goela_isit_2014} that the polarization process associated with $\rho^{2}_{\rm{max}}(W)$ converges almost surely on $[0, 1]$ to the endpoints of the interval. However, the supermartingale property of the random process as given in Theorem~\ref{thm:MaximalCorrelationSuperMG} was not established in prior work. }
\end{remark}

\begin{remark}[Asymptotic Analysis] \em{ The Bhattacharyya parameter $Z(W)$ has been instrumental to the analysis of the rate of polarization and scaling laws~\cite{telatar_rate_polar_2009,hassani_rate_dependent_2013,hassani_finite_scaling_2014,guruswami_gap_to_capacity_2015,mondelli_finite_scaling_2016}. Potentially, tracking the second-order moment $M_2(W)$ or channel dispersion $V(W)$ could lead to a refined asymptotic analysis. Further research is also necessary to determine whether $\rho^{2}_{\rm{max}}(W)$ plays a role in the asymptotic analysis of multi-level polarization or polarization of correlated random variables. }
\end{remark}

\begin{remark}[Channel-valued random processes] \em{ As discussed in Sec.~\ref{subsec:RelatedPriorWork}, a general approach considers \emph{channel-valued} polarization processes. Following our work, \cite{raj_nasser_2019} discusses the convergence of \emph{channel-valued} random processes for which a suitable topological space must be defined. }
\end{remark}

\section{Conclusion}

The mathematical results of Theorem~\ref{thm:fRelationsForOutputSymmetricBICs}, Corollary~\ref{corollary:NewMartingales} and Theorem~\ref{thm:MaximalCorrelationSuperMG} reflect both the algebraic and probabilistic foundation of channel polarization. For the broad class of symmetric $\BIC$s, we have shown that the random processes associated with $\I_f(W)$ for various choices of $f$ satisfy the martingale conditions for convergence if and only if specific \emph{functional inequalities} hold. The inequalities hold for specific convex and non-convex $f$. As summarized in Table~\ref{tbl:FunctionalsOfBICs}, we have developed a framework relying on the Blackwell representation of channels to analyze a variety of channel functionals, and their associated real-valued random processes. Further analysis is warranted to study \emph{channel-valued} random processes.



%

\appendices

\section{Proof of Eqn.~\eqref{eqn:BayesianInformationClaim}}\label{app:BayesianInformationGainProof}
To prove the claim, we first write down a closed-form expression for $B_\lambda(W)$. For any decoder $g$,
\begin{align*}
& {\mathbf P}[g(Y)\neq X] \notag \\
&= \E[\E[{\bf 1}_{\{g(Y) \neq X\}}|Y]] \\
&= \E\left[P_{X|Y}(0|Y){\bf 1}_{\{g(Y) = 1\}} + P_{X|Y}(1|Y) {\bf 1}_{\{g(Y) = 0\}}\right] \\
&= \E\Biggl[ \frac{\bar{\lambda}W(Y|0)}{\bar{\lambda}W(Y|0) + \lambda W(Y|1)} {\bf 1}_{\{g(Y) = 1\}} \notag \\
& \quad + \frac{\lambda W(Y|1)}{\bar{\lambda}W(Y|0)+\lambda W(Y|1)}{\bf 1}_{\{g(Y) = 0\}}\Biggl] \\
&= \sum_{y \in \sY} \left( \bar{\lambda}W(y|0){\bf 1}_{\{g(y)=1\}} + \lambda W(y|1){\bf 1}_{\{g(y) = 0\}}\right),
\end{align*}
and the minimum over all $g$ is evidently achieved by
\begin{align}\label{eq:g_opt}
g^*(y) \deq \begin{cases} 1, & \text{if } \lambda W(y|1) \ge \bar\lambda W(y|0) \\
0, & \text{if } \lambda W(y|1) < \bar\lambda W(y|0)
\end{cases}.
\end{align}
This yields
\begin{align}
b_\lambda(W) = \sum_{y \in \sY} (\bar \lambda W(y|0)) \wedge (\lambda W(y|1)).
\end{align}
In particular, $b_\lambda(\BSC(1/2)) = \bar\lambda \wedge {\lambda}$. Moreover, using the identity $a \wedge b = \frac{1}{2}(a+b-|a-b|)$, we can write
\begin{align*}
& B_\lambda(W) \notag \\
&= \bar \lambda \wedge \lambda -  \sum_{y \in \sY} (\bar \lambda W(y|0)) \wedge (\lambda W(y|1)) \\
&= \frac{1}{2}\Biggl[ 1-|1-2\lambda| - \notag \\
& ~~\sum_{y \in \sY} \left( \bar\lambda W(y|0) + \lambda W(y|1) - |\bar\lambda W(y|0) - \lambda W(y|1)|\right)\Biggl] \\
&= \frac{1}{2} \sum_{y \in \sY}|\bar\lambda W(y|0) - \lambda W(y|1)| - \frac{1}{2}|1-2\lambda|.
\end{align*}
We are now ready to prove the claim that $B_\lambda(W) = \I_{f_\lambda}(W)$. To that end,  consider a random couple $(X,Y)$ with $X \sim \Bernoulli(1/2)$ and $P_{Y|X} = W$ and let $S = \BlkFn_{W}(Y)$. Then, using the fact that $\E[S] = 1/2$, we have
\begin{align*}
& \I_{f_\lambda}(W) \notag \\
&= \E[f_\lambda(S)] \\
&= \frac{1}{2}(1-|1-2\lambda|) - \E[\bar\lambda S + \lambda \bar{S} - |\bar\lambda S - \lambda \bar{S}|] \\
&=\E[|\bar\lambda S - \lambda \bar S|] - \frac{1}{2}|1-2\lambda| \\
&= \frac{1}{2}\E\left[ \frac{1}{P_Y(Y)}\left|\bar\lambda W(Y|0)- \lambda W(Y|1)\right| \right] - \frac{1}{2}|1-2\lambda|  \\
&= \frac{1}{2}\sum_{y \in \sY} |\bar\lambda W(y|0) - \lambda W(y|1)| - \frac{1}{2}|1-2\lambda| \\
&\equiv B_\lambda(W).
\end{align*}
In particular, when $\lambda = 1/2$, the optimal decoder in \eqref{eq:g_opt} reduces to the maximum-likelihood (ML) rule, and the Bayesian information gain may be written in simplified form as follows,
\begin{align*}
B_{1/2}(W) = \frac{1}{4} \sum_{y \in \sY} |W(y|0) - W(y|1)|.
\end{align*}
In that case, $f_{1/2}(s) = \frac{1}{2} - s \wedge \bar{s} = \frac{1}{2} - \frac{1}{2}(1 - |2s-1|) = \frac{1}{2}|2s-1|$, and therefore
\begin{align*}
B_{1/2}(W) = \frac{1}{2} - P_{{\rm e,ML}}(W),
\end{align*}
where $P_{{\rm e,ML}}(W)$ denotes the probability of error of maximum-likelihood decoding of a single equiprobable bit sent through the channel $W$ \cite[Eqn.~1.9, Ch.~5]{alsan_phd14}. This, in turn, shows that $1-2P_{{\rm e,ML}}(W) = \I_f(W)$ with $f(s) = |2s-1|$.

\section{Proof of Lemma~\ref{lemma:TransitionMatBSC}}\label{appendix:TransitionMatPolarizedBSC}

Consider the two original channels $W_1 \equiv \BSC(\GAM)$ and $W_2 \equiv \BSC(\NU)$. The polar transform yields $W_1 \ATbad W_2 \equiv \BSC(\GAM) \ATbad \BSC(\NU)$ and $W_1 \ATgood W_2 \equiv \BSC(\GAM) \ATgood \BSC(\NU)$.

The output alphabet of $W_1 \ATbad W_2$ is $\{0, 1\} \times \{0, 1\}$ with conditional distribution $(W_1 \ATbad W_2)(y_1, y_2 | u_1)$. The conditional probabilities given an input $u_1 = 0$ are as follows: $(W_1 \ATbad W_2)(0,0|0)$ $=$ $(W_1 \ATbad W_2)(1,1|0)$ $=$ $\frac{1}{2}(1 - p\star q)$; $(W_1 \ATbad W_2)(0,1|0)$ $=$ $(W_1 \ATbad W_2)(1,0|0)$ $=$ $\frac{1}{2}(p\star q)$. Similarly, the conditional probabilities for a binary input $u_1 = 1$ are as follows: $(W_1 \ATbad W_2)(0,0|1)$ $=$ $(W_1 \ATbad W_2)(1,1|1)$ $=$ $\frac{1}{2}(p\star q)$; $(W_1 \ATbad W_2)(0,1|1)$ $=$ $(W_1 \ATbad W_2)(1,0|1)$ $=$ $\frac{1}{2}(1 - p\star q)$. Consider the following \emph{disjoint} sets of output pairs,
\begin{IEEEeqnarray}{rCL}
\mathcal{S^{-}} & = & \{(0,0),(1,1)\}, \nonumber \\
\mathcal{T^{-}} & = & \{(0,1),(1,0)\}. \nonumber
\end{IEEEeqnarray}
The union $\mathcal{S^{-}} \cup \mathcal{T^{-}}$ contains all $4$ output pairs. Viewing all output pairs grouped in each set $\mathcal{S^{-}}$ and $\mathcal{T^{-}}$ as super-symbols, the transition matrix $\tilde{T}_{W_1 \ATbad W_2}$ is as claimed.


The proof regarding $W_1 \ATgood W_2$ follows in an identical manner. The output alphabet of $W_1 \ATgood W_2$ is $\{0, 1\} \times \{0, 1 \} \times \{0, 1\}$ with conditional distribution $(W_1 \ATgood W_2)(y_1, y_2, u_1 | u_2)$. The conditional probabilities for a binary input $u_2 = 0$ are as follows: $(W_1 \ATgood W_2)(0,0,0|0)$ $=$ $(W_1 \ATgood W_2)(1,0,1|0)$ $=$ $\frac{1}{2}\bar{p}\bar{q}$; $(W_1 \ATgood W_2)(0,1,1|0)$ $=$ $(W_1 \ATgood W_2)(1,1,0|0)$ $=$ $\frac{1}{2} pq $; $(W_1 \ATgood W_2)(0,0,1|0)$ $=$ $(W_1 \ATgood W_2)(1,0,0|0)$ $=$ $\frac{1}{2}p\bar{q}$; $(W_1 \ATgood W_2)(0,1,0|0)$ $=$ $(W_1 \ATgood W_2)(1,1,1|0)$ $=$ $\frac{1}{2}\bar{p}q $. Similarly, the conditional probabilities for a binary input $u_2 = 1$ are given by: $(W_1 \ATgood W_2)(0,0,0|1)$ $=$ $(W_1 \ATgood W_2)(1,0,1|1)$ $=$ $\frac{1}{2} pq $; $(W_1 \ATgood W_2)(0,1,1|1)$ $=$ $(W_1 \ATgood W_2)(1,1,0|1)$ $=$ $\frac{1}{2} \bar{p}\bar{q} $; $(W_1 \ATgood W_2)(0,0,1|1)$ $=$ $(W_1 \ATgood W_2)(1,0,0|1)$ $=$ $\frac{1}{2} \bar{p}q $; $(W_1 \ATgood W_2)(0,1,0|1)$ $=$ $(W_1 \ATgood W_2)(1,1,1|1)$ $=$ $\frac{1}{2} p\bar{q} $. Consider the following \emph{disjoint} sets of output pairs,
\begin{IEEEeqnarray}{rCL}
\mathcal{S^{+}} & = & \{(0,0,0), (1,0,1)\}, \nonumber \\
\mathcal{T^{+}} & = & \{(0,1,1), (1,1,0)\}, \nonumber \\
\mathcal{B^{+}} & = & \{(0,0,1), (1,0,0)\}, \nonumber \\
\mathcal{G^{+}} & = & \{(0,1,0), (1,1,1)\}. \nonumber
\end{IEEEeqnarray}
The union $\mathcal{S^{+}} \cup \mathcal{T^{+}} \cup \mathcal{B^{+}} \cup \mathcal{G^{+}}$ contains all $8$ output pairs. Viewing all output pairs in the sets $\mathcal{S^{+}}$, $\mathcal{T^{+}}$, $\mathcal{B^{+}}$ and $\mathcal{G^{+}}$ as super-symbols, the transition matrix $\tilde{T}_{W_1 \ATgood W_2}$ is as claimed.

The parallel broadcast channel $W_1 \times W_2 \equiv \BSC(\GAM) \times \BSC(\NU)$ has an output alphabet $\{0, 1\} \times \{0,1\}$ with conditional distribution denoted as $(W_1 \times W_2)(y_1, y_2 | x)$. The conditional probabilities for binary input $x = 0$ are as follows: $(W_1 \times W_2)(0,0|0) = \bar{p}\bar{q}$; $(W_1 \times W_2)(0,1|0) = \bar{p}q$; $(W_1 \times W_2)(1,0|0) = p\bar{q}$; $(W_1 \times W_2)(1,1|0) = pq$. Similarly the conditional probabilities for binary input $x = 1$ are as follows: $(W_1 \times W_2)(0,0|1) = pq$; $(W_1 \times W_2)(0,1|1) = p\bar{q}$; $(W_1 \times W_2)(1,0|1) = \bar{p}q$; $(W_1 \times W_2)(1,1|1) = \bar{p}\bar{q}$. By comparing the transition probabilities, it is evident that $T_{\BSC(\GAM) \times \BSC(\NU)} = \tilde{T}_{\BSC(\GAM) \ATgood \BSC(\NU)}$ as claimed.


\section{Proof of Lemma~\ref{lemma:StructuralPolarizedBSC}}\label{appendix:StructuralPolarizedBSC}

Eqn.~\eqref{eqn:StructuralDecompPolarizedBSCATbad} follows directly from the transition matrix Eqn.~\eqref{eqn:TransitionMatATbad} of Lemma~\ref{lemma:TransitionMatBSC}. To see why Eqn.~\eqref{eqn:StructuralDecompPolarizedBSCATgood} holds, consider the transition matrix of Eqn.~\eqref{eqn:TransitionMatATgood}. Assume $(\GAM \star \NU) \neq 0$. Let error probabilities $\alpha$ and $\beta$  be defined as in Eqn.~\ref{eqn:AlphaNewBSCErrorProb} and Eqn.~\ref{eqn:BetaNewBSCErrorProb} respectively.
\begin{IEEEeqnarray}{rCL}
\tilde{T}_{\BSC(\GAM) \ATgood \BSC(\NU)} & = & T_{\BSC(\GAM) \times \BSC(\NU)} \nonumber \\
& = & \left[ \begin{array}{cccc} ~\bar{\GAM} \bar{\NU} & ~\GAM \NU & ~\GAM \bar{\NU} & ~\bar{\GAM} \NU \\ ~\GAM \NU & ~\bar{\GAM}\bar{\NU} & ~\bar{\GAM} \NU & ~\GAM \bar{\NU} \end{array}\right] \nonumber \\
& = & \left[ \begin{array}{cc} (1 \! - \! \GAM \star \NU) \bar{\alpha} ~ & ~ (1 \! - \! \GAM \star \NU) \alpha \\ (1 \! - \! \GAM \star \NU) \alpha ~&~ (1 \! - \! \GAM \star \NU) \bar{\alpha} \\ (\GAM \star \NU)\bar{\beta} ~&~ (\GAM \star \NU)\beta \\ (\GAM \star \NU)\beta ~&~ (\GAM \star \NU)\bar{\beta} \end{array} \right]^{T}. \nonumber
\end{IEEEeqnarray}

The above transition matrix for $\BSC(\GAM) \ATgood \BSC(\NU)$ reveals the structural decomposition of the one-step polarized channel as established by Theorem~\ref{thm:BIOSdecomp}. More precisely, $\BSC(\GAM) \ATgood \BSC(\NU)$ is a $\BSC(\alpha)$ with probability $(1 - \GAM \star \NU)$ and a $\BSC(\beta \wedge \bar{\beta})$ with probability $(\GAM \star \NU)$. The transformed error probabilities are specified so that $\alpha \in [0,\frac{1}{2}]$ and $(\beta \wedge \bar{\beta}) \in [0, \frac{1}{2}]$.


\section{Proof of Corollary~\ref{corollary:BlackwellMeasuresPolarizedBSC}}\label{appendix:BlackwellMeasuresPolarizedBSC}

The corollary holds for $(\GAM, \NU) = (0,0)$ trivially as a degenerate case. Therefore, assume $(\GAM, \NU) \neq (0,0)$ so that $\GAM \star \NU \in (0, \frac{1}{2}]$. From Theorem~\ref{thm:BM_polarization},
\begin{IEEEeqnarray}{rCl}
\sm_{\BSC(\GAM) \ATbad \BSC(\NU)} & = & \sm_{\BSC(\GAM)} \ATbad \sm_{\BSC(\NU)}, \nonumber \\
\sm_{\BSC(\GAM) \ATgood \BSC(\NU)} & = & \sm_{\BSC(\GAM)} \ATgood \sm_{\BSC(\NU)}, \nonumber
\end{IEEEeqnarray}
where the operations $\ATbad$ and $\ATgood$ on Blackwell measures were defined in Definition~\ref{def:ATbadATgoodProbMeasures}. Thus, consider two independent random variables $S_1 \sim \sm_{\BSC(\GAM)}$ and $S_2 \sim \sm_{\BSC(\NU)}$. The random variable $S_1$ takes two equiprobable values $\GAM$ and $\bar{\GAM}$, and $S_2$ takes two equiprobable values $\NU$ and $\bar{\NU}$.

To prove Eqn.~\eqref{eqn:corollary:BlackwellMeasureBadPolarizedBSC}, consider Eqn.~\eqref{eq:bad_BM} of Definition~\ref{def:ATbadATgoodProbMeasures}. The random variable $1-S_1 \star S_2$ takes two equiprobable values, $\GAM \star \NU$ and $1-\GAM \star \NU$. The integral of Eqn.~\eqref{eq:bad_BM} may be evaluated for any continuous  $f : [0,1] \to \Reals$ as follows:
\begin{IEEEeqnarray}{rCl}
\int_{[0,1]}f \d\sm_{\BSC(\GAM) \ATbad \BSC(\NU)} & = & \frac{1}{2}f\left(\GAM \star \NU\right) + \frac{1}{2}f\left(1 \! - \! \GAM \star \NU\right). \nonumber
\end{IEEEeqnarray}
The corresponding Blackwell measure of $\BSC(\GAM) \ATbad \BSC(\NU)$ written as a weighted sum of Dirac measures is,
\begin{IEEEeqnarray}{rCl}
\sm_{\BSC(\GAM) \ATbad \BSC(\NU)} & = & \frac{1}{2}\delta_{\GAM \star \NU} + \frac{1}{2}\delta_{1-\GAM \star \NU}. \nonumber
\end{IEEEeqnarray}

To prove Eqn.~\eqref{eqn:corollary:BlackwellMeasureGoodPolarizedBSC}, consider Eqn.~\eqref{eq:good_BM} of Definition~\ref{def:ATbadATgoodProbMeasures}. The integral of Eqn.~\eqref{eq:good_BM} may be evaluated for any continuous $f : [0,1] \to \Reals$ as follows:
\begin{IEEEeqnarray}{rCcl}
\int_{[0,1]}f \d\sm_{\BSC(\GAM) \ATgood \BSC(\NU)} & = & \Biggl( & \frac{1}{2} (1 \! - \! \GAM \star \NU) f\left(\frac{\GAM \NU}{1 \! - \! \GAM \star \NU}\right) \nonumber \\
& & + & \frac{1}{2}(1 \! - \! \GAM \star \NU) f\left(\frac{\bar{\GAM}\bar{\NU}}{1 \! - \! \GAM \star \NU}\right) \nonumber \\
& & + & \frac{1}{2} (\GAM \star \NU) f\left(\frac{\GAM \bar{\NU}}{\GAM \star \NU}\right) \nonumber \\
& & + & \frac{1}{2}(\GAM \star \NU) f\left( \frac{\bar{\GAM}\NU}{\GAM \star \NU}\right) \Biggl). \nonumber
\end{IEEEeqnarray}
The corresponding Blackwell measure of $\BSC(\GAM) \ATgood \BSC(\NU)$ written as a weighted sum of Dirac measures is
\begin{IEEEeqnarray}{rlCl}
\sm_{\BSC(\GAM) \ATgood \BSC(\NU)} = (1-\GAM \star \NU ) \Biggl( & \frac{1}{2}\delta_{\frac{\GAM \NU}{1- \GAM \star \NU}} & + & \frac{1}{2}\delta_{\frac{\bar{\GAM}\bar{\NU}}{1- \GAM \star \NU}}\Biggl) \nonumber \\
\quad +\> (\GAM \star \NU ) \Biggl( & \frac{1}{2}\delta_{\frac{ \GAM \bar{\NU}}{\GAM \star \NU}} & + & \frac{1}{2}\delta_{\frac{\bar{\GAM}\NU}{\GAM \star \NU}}\Biggl). \nonumber
\end{IEEEeqnarray}
As shown in the proof of Lemma~\ref{lemma:TransitionMatBSC} by direct computation, the above probability measure $\sm_{\BSC(\GAM) \ATgood \BSC(\NU)}$ is equivalent to the probability measure $\sm_{\BSC(\GAM) \times \BSC(\NU)}$. In addition, Eqn.~\eqref{eqn:corollary:BlackwellMeasureGoodParallelBroadcast} follows from Eqn.~\eqref{eqn:StructuralDecompPolarizedBSCATgood} of Lemma~\ref{lemma:StructuralPolarizedBSC}.


\section{Proof of Corollary~\ref{corollary:BIOSPolarization}}\label{appendix:BlackwellMeasuresPolarizedSymmetricBICs}

From Theorem~\ref{thm:BIOSdecomp}, there exists a structural decomposition for output-symmetric $\BIC$s $W_1$ and $W_2$. The Blackwell measures of $W_1$ and $W_2$ may be written as follows:
\begin{IEEEeqnarray}{rClCl}
\sm_{W_1} & = & \sum^m_{i=1}\lambda_i \sm_{\BSC(\GAM_i)} & = & \sum^m_{i=1} \frac{\lambda_i(\delta_{\GAM_i}+\delta_{\bar{\GAM}_i})}{2}, \nonumber \\
\sm_{W_2} & = & \sum^k_{j=1}\mu_j \sm_{\BSC(\NU_j)} & = & \sum^k_{j=1} \frac{\mu_j(\delta_{\NU_j}+\delta_{\bar{\NU}_j})}{2}, \nonumber
\end{IEEEeqnarray}
for some choices of parameters $(m,\lambda, \GAM)$ and $(k,\mu, \NU)$. Therefore, for two independent r.v.'s $S_1 \sim \sm_{W_1}$ and $S_2 \sim \sm_{W_2}$ and for any continuous $f : [0,1] \to \Reals$ we have
\begin{IEEEeqnarray}{rCl}
\int_{[0,1]} f\d\sm_{W_1 \ATbad W_2} & = & \E[f(1-S_1 \star S_2)] \nonumber \\
\IEEEeqnarraymulticol{2}{r} { = \sum^m_{i=1}\sum^k_{j=1}\lambda_i \mu_j } & \Biggl( \frac{1}{2} f(\GAM_i \star \NU_j) + \frac{1}{2}f(1-\GAM_i \star \NU_j)\Biggl). \nonumber \\
\noalign{\vspace{1.5\jot} \noindent In addition, we have
\vspace{1.5\jot}}
\int_{[0,1]}f\d\sm_{W_1 \ATgood W_2} & = & \E\Bigg[(1-S_1 \star S_2)f\left(\frac{S_1S_2}{1-S_1 \star S_2}\right) \nonumber \\
& & \qquad \qquad +~ (S_1 \star S_2)f\left(\frac{\bar{S}_1S_2}{S_1 \star S_2}\right)\Bigg] \nonumber \\
\IEEEeqnarraymulticol{2}{r} { = \sum^m_{i=1}\sum^k_{j=1}\lambda_i \mu_j } & \Biggl( \frac{1}{2} (1-\GAM_i \star \NU_j) f\left(\frac{\GAM_i \NU_j}{1-\GAM_i \star \NU_j}\right)  \nonumber \\
\IEEEeqnarraymulticol{2}{r} { \quad } & + \frac{1}{2}(1-\GAM_i \star \NU_j) f\left(\frac{\bar{\GAM}_i\bar{\NU}_j}{1-\GAM_i \star \NU_j}\right)  \nonumber \\
\IEEEeqnarraymulticol{2}{r} { \quad } & + \frac{1}{2} (\GAM_i \star \NU_j) f\left(\frac{\GAM_i\bar{\NU}_j}{\GAM_i \star \NU_j}\right)   \nonumber \\
\IEEEeqnarraymulticol{2}{r} { \quad } & + \frac{1}{2}(\GAM_i \star \NU_j) f\left( \frac{\bar{\GAM}_i \NU_j}{\GAM_i \star \NU_j}\right)  \Biggl). \nonumber
\end{IEEEeqnarray}
Applying Corollary~\ref{corollary:BlackwellMeasuresPolarizedBSC}, we obtain Eqns.~\eqref{eq:BIOSbad} and \eqref{eq:BIOSgood}.

\section{Proof Of Theorem~\ref{thm:MaximalCorrelationSuperMG}}\label{appendix:ProofNewSuperMGMaximalCorrelations}

We prove the functional inequality by expanding and simplifying the gap function $\textsc{gap}_{f}(p, q)$ where $f : [0,1] \to \Reals$ is chosen to be the polynomial function $f(s) = (2s-1)^{2}$. The terms given in Eqn.~\eqref{eqn:GapFunctionMaxCorrDefinition} are expanded as follows:
\begin{IEEEeqnarray}{rCl}
f(p) & = & 4p^{2} - 4p + 1. \nonumber \\
f(q) & = & 4q^{2} - 4q + 1. \nonumber \\
f(p \star q) & = & 16pq\bar{p}\bar{q} + 4p^{2} + 4q^{2} - 4p - 4q + 1.  \nonumber \\
(1 - p\star q)f(\alpha) & = & \frac{4p^{2}q^{2}}{1 - p \star q} - 4pq + (1 - p\star q). \label{eqn:MaxCorrAppendixFirstAlphaTerm}  \\
(p \star q)f(\beta) & = & \frac{4\bar{p}^{2}q^{2}}{p \star q} - 4\bar{p}q + (p \star q). \label{eqn:MaxCorrAppendixSecondBetaTerm}  \\
\noalign{\vspace{1.5\jot} \noindent To compute Eqn.~\eqref{eqn:MaxCorrAppendixFirstAlphaTerm} and Eqn.~\eqref{eqn:MaxCorrAppendixSecondBetaTerm}, recall that the parameters $\alpha := \frac{ \GAM \NU }{ 1 - \GAM \star \NU }$ and $\beta := \frac{ \bar{\GAM} \NU }{ \GAM \star \NU }$ were defined in Lemma~\ref{lemma:StructuralPolarizedBSC}. Since $f(s) = f(\bar{s})$, the gap function is given by Eqn.~\eqref{eqn:DefinitionGapFunctionCorollary}. After a series of algebraic simplifications,
\vspace{1.5\jot}}
\textsc{gap}_{f}(p, q) & = & f(p) + f(q) - f(p \star q) \nonumber \\
&& -\> (1 - p \star q)f(\alpha) - (p \star q)f(\beta) \nonumber \\
& = & -16pq\bar{p}\bar{q} -\frac{4p^{2}q^{2}}{1 - p \star q} -\frac{4\bar{p}^{2}q^{2}}{p \star q} + 4q  \IEEEeqnarraynumspace \nonumber \\
& = & -16pq\bar{p}\bar{q} -\frac{4p^{2}q^{2}}{1 - p \star q} + \frac{4pq(1-pq)}{p \star q}  \IEEEeqnarraynumspace \nonumber \\
& = & -16pq\bar{p}\bar{q} + \frac{4pq\bar{p}\bar{q}}{(1 - p\star q)(p \star q)} \nonumber \\
& = & \frac{4pq\bar{p}\bar{q}(2p-1)^{2} (2q-1)^{2}}{(1 - p \star q)(p \star q)}. \nonumber
\end{IEEEeqnarray}
For all $p, q \in [0, \frac{1}{2}]$, $(p, q) \neq (0, 0)$, $\textsc{gap}_{f}(p, q) = \textsc{gap}_{f}(q, p)$, and $\textsc{gap}_{f}(p, q) \geq 0$. This concludes the verification of Eqn.~\eqref{eqn:GapFunctionMaxCorrelation} in the proof of Theorem~\ref{thm:MaximalCorrelationSuperMG}.

%

\section*{Acknowledgment}

The authors thank Dr.\ Mine Alsan, Prof.\ Venkat Anantharam, and Prof.\ Erdal \Arikan\ for helpful discussions. The authors also thank the associate editor and four anonymous reviewers whose valuable comments helped improve this manuscript.

\ifCLASSOPTIONcaptionsoff
  \newpage
\fi



\bibliographystyle{IEEEtran}
\bibliography{IEEEabrv,./main_journal_bibliography}

\vfill

\enlargethispage{-5in}

\end{document}